\documentclass[a4paper, 11pt]{article}

\usepackage{amssymb,amsmath,amsthm,longtable,verbatim}

\usepackage{graphicx}

\usepackage[margin=3cm]{geometry}

\numberwithin{equation}{section}

\newcommand{\D}{{\mathbb D}}
\newcommand{\C}{{\mathbb C}}
\newcommand{\A}{{\mathbb A}}
\newcommand{\Z}{{\mathbb Z}}\newcommand{\E}{{\mathbb E}}
\newcommand{\bS}{{\mathbb S}}
\renewcommand{\H}{{\mathbb H}}

\newcommand{\crs}{\beta}
\renewcommand{\P}{{\mathbb P}}

\newcommand{\R}{{\mathbb R}}

\newcommand{\const}{\mathrm{const}\, }
\newcommand{\comp}{\mathcal{C}}

\newcommand{\SLE}{\mathrm{SLE}}

\newcommand{\mesh}{\delta}
\newcommand{\Od}{\Omega^\mesh}

\newcommand{\Cvr}{\varpi}

\newcommand\Conf{\mathrm{Conf}}

\newcommand\wind{\mathrm{w}}
\newcommand\pa{\partial}

\renewcommand\Re{\mathrm{Re}}
\renewcommand\Im{\mathrm{Im}}

\def\CSpace{{\Theta}}

\def\res{\mathop{\mathrm{res}}\limits}

\theoremstyle{remark}

\theoremstyle{plain}
\newtheorem{theorem}{Theorem}[section]

\newtheorem{lemma}[theorem]{Lemma} 
\newtheorem{corollary}[theorem]{Corollary} 
\newtheorem{proposition}[theorem]{Proposition} 

\theoremstyle{remark}
\newtheorem{rem}[theorem]{Remark}

\date{}

\begin{document}
\title{Critical Ising interfaces in multiply-connected domains}
\author{Konstantin Izyurov}
\maketitle
\begin{abstract}
We prove a general result on convergence of interfaces in the critical planar Ising model to conformally invariant curves absolutely continuous with respect to SLE(3). Our setup includes multiple interfaces on arbitrary finitely connected domains, and we also treat the radial SLE case. In the case of simply and doubly connected domains, the limiting processes are described explicitly in terms of rational and elliptic functions, respectively. 
\end{abstract}
\newcounter{Listcounter}
\section{Introduction}
In 2000, Schramm \cite{Schramm00} proposed his celebrated tool to study conformally invariant two-dimensional lattice models, namely, the Stochastic Loewner Evolutions (SLE). These are random conformally invariant planar curves that describe limits of interfaces, or discrete curves, in lattice models of statistical physics at criticality. Numerous conjectures have been made on convergence of interfaces to $\SLE$, and a few have actually been proven  rigorously \cite{LSW, Smirnov01, SS_harm, SS_DGFF, CHS3}.  We refer the reader to \cite{Lawler-book, Werner-book} for a background on Schramm-Loewner evolution.

The Ising model is one of the most well-studied models of statistical mechanics, and its conformal invariance at criticality has long been conjectured by physicists. In 2006, S. Smirnov proved this conjecture for fermionic observables \cite{Smirnov06, CHS2}. This, in particular, led to the  proof of convergence of the chordal Ising and FK-Ising interfaces to chordal SLE${}_3$ and SLE${}_\frac{16}{3}$ curves, respectively \cite{CHS3, SmirnovAntti}. In the latter case a description of the full scaling limit of loop soups in terms of the SLE($\frac{16}{3}$,$-\frac{2}{3}$) process was also announced. D. Chelkak and S. Smirnov \cite{CHS2} also proved universality of the critical Ising model on isoradial graphs, and C. Hongler and K. Kyt\"ol\"a extended convergence of the interface to the case $+/-/$free boundary conditions, where the scaling limit is dipolar SLE${}_3$ \cite{HonKyt}. Another proof of this result, and a generalization to arbitrary number of marked points, was recently proposed 
by the author \cite{Izy_free}.

Most of the SLE literature deals with the \emph{chordal} SLE in simply-connected domains; SLE in multiply-connected domains and multiple SLE are less studied. The difficulties arising in these cases are both conceptual, e. g. there is no unique definition of SLE because of lack of Schramm's principle, and technical, since explicit computations are hard. Nevertheless, D. Zhan \cite{Zhan-LERW} was able to extend the convergence of planar loop-erased random walks \cite{LSW} to the case of multiply-connected domains. M. Bauer, D. Bernard and K. Kyt\"ol\"a \cite{BBK05} suggested to define multiple SLE using Conformal Field Theory correlation functions. G. Lawler and M. Kozdron \cite{LK} and J. Dub\'edat \cite{Dub07} obtained similar definitions by imposing certain consistency (or ``commutation'') conditions. A definition of SLE for arbitrary $\kappa$ in multiply-connected domains on the grounds of restriction properties of SLE was proposed by G. Lawler in \cite{Law11}. Recently, in a series of papers \cite{FK}, S.
 Flores and P. Kleban have given 
a rigorous description of the space of multiple SLE partition functions. An alternative description was suggested by K. Kyt\"ol\"a and E. Peltola \cite{KP, KP15}.

\subsection{Statements of the main results}

The present paper aims at generalizing the convergence results of \cite{CHS3} for the interfaces in the Ising model. Let $\Od$ and $a_1^\delta,\dots,a_{2k}^\delta\in\partial\Od$ be discrete approximations (see Section \ref{Sec: notations} for precise definitions) to a finitely-connected domain $\Omega$ with marked points $a_1,\dots,a_{2k}\in\partial\Omega$. We will work with the \emph{low-temperature} expansion of the Ising model in $\Od$ with ``boundary change operators'' at $a^\delta_1,\dots,a^\delta_{2k}$. (Alternatively, it can be viewed as a dilute $O(1)$ loop model, with ``one-leg insertions'' on the boundary.) The \emph{configuration set}  $\Conf:=\Conf(\Od,a_1^\delta,\dots,a_{2k}^\delta)$ of the model consists of subsets $S$ of edges of $\Od$ such that any vertex of $\Od$ is incident to an even number of edges $S$, and $a_1^\delta,\dots,a^\delta_{2k}$ is the set of boundary edges in $S$, i. e. those connecting a vertex in $\Od$ with one not in $\Od$; see Figure \ref{fig: Conf}. The probability of a 
configuration $S$ is given by
$$
\P(S)=\frac{x^{|S|}}{Z},
$$
where $|S|$ is the number of edges in $S$ and $Z=\sum\limits_{S\in\Conf}x^{|S|}$ is the partition function. Throughout the paper, the weight $x$ will be set to its critical value:
$$
x=x_{crit}=\sqrt{2}-1.
$$
We discuss the relation of this setup to the spin Ising model below.

Any configuration $S\in\Conf(\Od,a_1^\delta,\dots,a_{2k}^\delta)$ can be decomposed into a collection of  simple, edge-disjoint loops and $k$ interfaces connecting the points $a^\delta_1\dots,a^\delta_{2k}$; the decomposition is non-unique, but the scaling limit is independent on its choice. We wish to describe the scaling limit of (initial segments of) the interface $\gamma=\gamma^\delta$ starting at the point $a^{\delta}_1$ on the boundary component $\partial_{a_1} \Od\ni a^\delta_1$. Note that we do not specify a priory at which point it ends up. We assume that $(\Od,a^\delta_1,\dots,a^\delta_{2k})$ converges to $\Omega,a_1,\dots,a_{2k}$ in the sense of Carath\'eodory (see, e.g., \cite{Pommerenke} for a background on this notion), and that $\partial_{a_1}\Omega$ consists of more than one point. We require the convergence to be regular near the marked points $(a_2,\dots,a_{2k})$ in the sense of \cite[Definition 3.14]{ChIz}, which means that in some small but fixed neighborhood of each $a_j$, the boundaries 
of $\Od$ are vertical or horizontal straight lines. With some work, this technical assumption can be 
removed, cf. \cite{Izy_free}.

\begin{figure}[t]
\begin{center}
\includegraphics[width=0.7\textwidth]{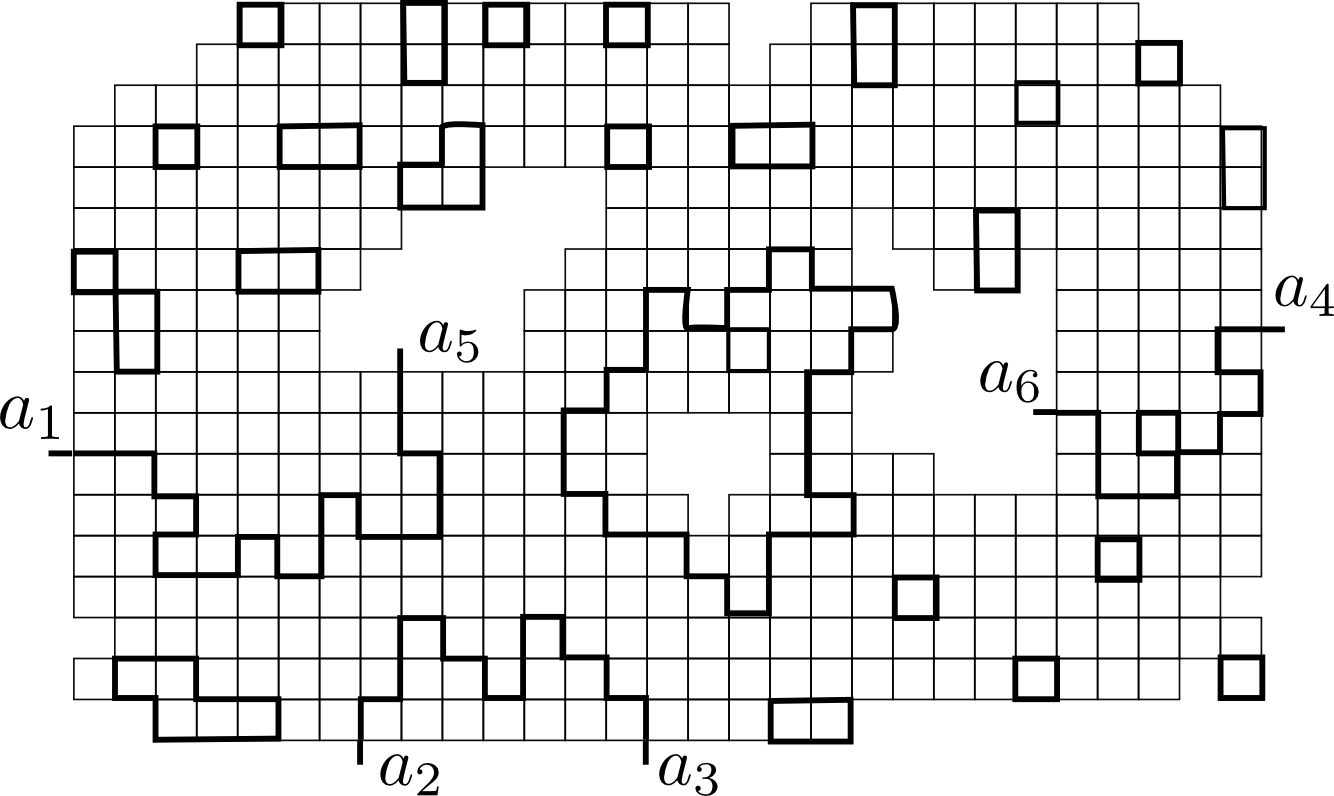}
\end{center}
\caption{\label{fig: Conf} An example of a four-connected discrete domains $\Od$ with six marked points and a configuration $S\in\Conf(\Od,a_1,\dots,a_{6})$}
\end{figure}

In order to describe the law of the interface in the multiply-connected setup, we follow \cite{Zhan-LERW} and use the standard chordal Loewner evolution in the upper half-plane $\H$. We call a finitely-connected domain $\Lambda$ \emph{an almost-$\H$ domain} if $\Lambda\subset \H$ and $\pa\Lambda$ consists of $\R$ and analytic Jordan arcs. Let $g_0$ be a conformal map from $\Omega$ to an almost-$\H$ domain $\Lambda_0$  such that $g_0(\partial_{a_1} \Omega)=\R$. For each $\delta$, we consider a conformal map $g_0^\delta$ from $\Od$ to an almost-$\H$ domain  $\Lambda^\delta_0$, such that as $\delta$ tends to zero, $g_0^\delta$ converges to $g_0$ uniformly on compact subsets of $\Omega$, hence also $g_0(a_m^\delta)\to g_0(a_m)$, and, moreover, $\partial \Lambda^\delta_0\to \partial \Lambda_0$ in the sense of~$C^2$. Such a sequence can be constructed in many ways, see Section~\ref{sec: proof}.

We encode the image $K_t^\delta:= g^\delta_0(\gamma^{\delta})$ by chordal Loewner evolution in $\H$, that is, parametrize $K_t^\delta$ by twice the half-plane capacity $t$ and consider the maps $g^\delta_t:=g_{K_t^\delta}\circ g_0^\delta$, where $g_{K_t^\delta}$ are conformal maps from $\H\backslash K_t^\delta$ to $\H$ with hydrodynamic normalization at infinity. These maps satisfy Loewner's equation
$$
\partial_t g^\delta_t(z)=\frac{2}{g^\delta_t(z)-a^{\delta}_1(t)},
$$
where $a^{\delta}_1(t)=g^\delta_t(\gamma_t)$ is a (random) continuous driving function which completely determines the curve $\gamma_t$. Our goal is to describe the law of $a^\delta_1(t)$ in the limit $\delta\to 0$.

The statements of our results will involve the \emph{SLE${}_3$ partition functions}, which are positive quantities $Z(\Omega,a_1,\dots,a_{2k})$ associated with each domain $\Omega$ with marked points $a_1,\dots,a_{2k} \in \partial \Omega$ such that $\pa \Omega$ is smooth near $a_1,\dots,a_{2k}$ (in particular, with every almost-$\H$ domain $\Lambda$).  We discuss this notion in detail in Section~\ref{sec: intro_pf}, for now we just mention that it depends nicely enough on $\Omega$ and $a_1,\dots,a_{2k}$. We define the \emph{drift term} by
$$
D(\Lambda,a_1,\ldots,a_{2k}):=3\partial_{a_1}\log Z(\Lambda,a_1,\dots,a_{2k})
$$
Given a Loewner chain $K_t$ in $\Lambda_0$ with a driving force $a_1(t)$, denote $\Lambda_t:=g_{K_t}(\Lambda_0)$, $a_i(t):=g_{K_t}(a_i(0))$ for $i=2,\dots,{2k}$. Let $B_t$ be a standard Brownian motion on $\R$, and let the driving force $a_1(t)$ be the strong solution to the integral equation
\begin{equation}
\label{eq: a_t}
a_1(t):=a_1(0)+\sqrt{3}B_t+\int\limits_0^tD (\Lambda_s,a_1(s),\dots,a_{2k}(s))ds
\end{equation}
with initial data $\Lambda_0=g_0(\Omega)$ and $a_i(0)=g_0(a_i)$. Note that $D$ is a function on the infinitely-dimensional space of almost-$\H$ domains. We will prove that this function is locally Lipschitz in  $L^{\infty}$ norm with respect to $a_1(\cdot)$. The existence and the uniqueness of a strong solution to (\ref{eq: a_t}) then can be derived in a standard way, see~\cite{Zhan-LERW}.

Fix any cross-cut $\crs$ in $\Lambda_0$ that separates $g_0(a_1)$ from other marked points, other boundary components and infinity. We denote by $T_\crs$ the first time that the Loewner chain driven by the solution $a_1(t)$ to (\ref{eq: a_t}) hits $\crs$. Similarly, denote by $T^\delta_\crs$ the first time that $K^\delta_t$ hits $\crs$. 

We will say that a random driving force $a_1^\delta(t)$ \emph{locally converges} to $a_1(t)$ if for any such crosscut $\crs$ the two processes can be coupled so that $\sup_{t\in [0,T_\crs^\delta\wedge T_\crs]}|a_1^\delta(t)-a_1(t)|$ tends to zero in probability as $\delta\to 0$.
\begin{theorem}
\label{thm: conv_int}
As $\delta\rightarrow 0$, the process $a^{\delta}_1(t)$, i. e. the random driving force of the conformal image of the discrete Ising interface, locally converges to $a_1(t)$, i. e. the unique strong solution to (\ref{eq: a_t}).
\end{theorem}

By now, enough tools have been developed allowing one to strengthen the topology of convergence:
\begin{corollary}
The convergence in Theorem \ref{thm: conv_int} holds in the $L^\infty$ norm on curves up to re-parametrization, as in \cite{CHS3}. 
\label{cor: stronger topology}
\end{corollary}
We sketch a proof in Section \ref{sec: proof}. Using the estimates of crossing probabilities obtained in \cite{CDH}, it should be possible to show that almost surely, as $t$ tends to $\sup_\crs T_\crs$, the tip of the curve driven by  $a_1(t)$ tends to one of the points $g_0(a_2),\ldots,g_0(a_{2k})$; see \cite{Disser} for weaker results in the same spirit.

In the simply-connected case, one has $\Lambda_t\equiv \H$, and thus the random driving force $a_1(t)$ is described in terms of a diffusion on a \emph{finitely dimensional} space. For a similar description in the doubly connected case, one has to use the annulus Loewner evolution \cite{BB-annulus, Zhan-annulus}. Recall that the annulus Loewner evolution is a solution to the differential equation
\begin{equation}
\partial_t g_t(z)=V^{p-t}_{\theta(t)}\left(g_t(z)\right);  
\label{eq: ann_Loewner}
\end{equation}
here $V^p_{\theta}(\cdot)$ is the Loewner kernel for the annulus $\A_{p}:=\{e^{-p}\leq |z|\leq 1\}$, and $\theta(t)\in\R$ is a driving function, see Section \ref{sec: annulus}. If the initial data $g_0(\cdot)$ is chosen to be a conformal map from a doubly connected domain $\Omega$ to $\A_p$, then $g_t(z)$ is a conformal map from a sub-domain of $\Omega$ to $\A_{p-t}$. 

It is convenient to write the partition function in terms of the arguments $\theta_{1},\ldots,\theta_{2k}$ of $a_1,\ldots,a_{2k}\in\pa\A_p$. By \emph{annulus SLE${}_3$} with partition function $Z=Z(p,\theta_{a_1},\dots,\theta_{a_{2k}})$, we mean the annulus Loewner evolution driven by the unique strong solution to the following SDE:
\begin{equation}
d\theta_{1}(t)=\sqrt{3}dB_t + D(p-t,\theta_{1}(t),\dots,\theta_{2k}(t))dt
\label{eqn: annulus}
\end{equation}
where  $D(p,\theta_1,\dots,\theta_{2k})= 3\partial_{\theta_{1}}\log Z(p,\theta_{1},\dots,\theta_{{2k}})$ and $\theta_i(t)=\arg g_t(a_i)$. 
\begin{figure}[t]
\begin{center}
\includegraphics[width=0.7\textwidth]{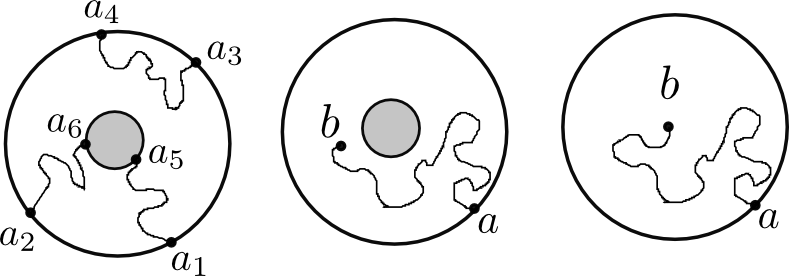}
\end{center}
\caption{\label{fig: Three cases} The setups covered by Propositions \ref{prop: ann1} (left) and \ref{prop: ann2}.}
\end{figure}

Similarly to the above setup, for each $\delta$ we choose a conformal map $g_0^\delta$ from $\Od$ to  $\A^{p_{\delta}}$, where $p_\delta$ is the conformal modulus of $\Od$ and $|g^\delta_0(a^\delta_1)|=1$, in such a way that $g_0^\delta$ converges to $g_0$ on compact subsets of $\Omega$. The curves $\gamma^\delta_{[0,t]}$ parametrized by the modulus of $\Od\backslash \gamma^\delta_{[0,t]}$ can be then described by an annulus Loewner evolution $g_t^\delta$ with a random driving process $\theta^\delta_{1}(t)=\arg a^{\delta}_1(t)$. Our result then reads as follows.
\begin{proposition}
\label{prop: ann1}
 As $\delta\to 0$, the process $\theta^\delta_1(t)$ above locally converges to the driving process of the annulus SLE with the partition function given by (\ref{eq: pf_ann_1})--(\ref{eq: pf_ann_2}) below.
\end{proposition}

It is also interesting to treat the following situation. Suppose $a^\delta$ is a boundary edge and $b^\delta$ is an \emph{inner vertex} in $\Od$, and denote by $\Conf(\Od,a^\delta,b^\delta)$ the set of subsets $S$ of edges in $\Omega$ such that $a^\delta$ is the only boundary edge in $S$ and all inner vertices have even degrees in $S$, except for $b^\delta$. We assume that the triplet $(\Od,a^\delta,b^\delta)$ approximates $(\Omega,a,b)$, where $a\in\pa\Omega$ and $b\in\Omega$. 

\begin{proposition}
\label{prop: ann2}
 In the case of a simply-connected (respectively, doubly-connected) domain $\Omega$, the interface connecting $a^\delta$ to $b^\delta$ converges as $\delta\to 0$ to the radial SLE${}_3$ (respectively, to the annulus SLE${}_3$ with partition function given by (\ref{eq: pf_annulus_radial}) below) in the sense of local convergence of driving functions.
\end{proposition}

Assume that, in the setup of Theorem \ref{thm: conv_int}, every boundary component contains an \emph{even} number of marked points $a_1,\ldots,a_{2k}$. Then, given $S\in\Conf(\Od,a_1^\delta,\ldots,a_{2k}^\delta)$, one can assign spins $\pm 1$ to the faces of $\Od$ so that the edges of $S$ separate faces with different spins. This gives a measure-preserving mapping of the low-temperature expansion to the critical spin Ising model with \emph{locally monochromatic} boundary conditions, i. e. with the spins conditioned to be constant ($+1$ or $-1$) along each boundary arc and changing at the marked points. Perhaps, more natural boundary conditions in the spin Ising setup are when these spins are \emph{fixed}. Since the spins change at $a_i^\delta$, this amounts to prescribing a spin at one face of each boundary component. Our techniques also cover this case modulo a result that will appear in a forthcoming paper with D. Chelkak and C. Hongler. 

\begin{figure}[t]
\begin{center}
\includegraphics[width=0.7\textwidth]{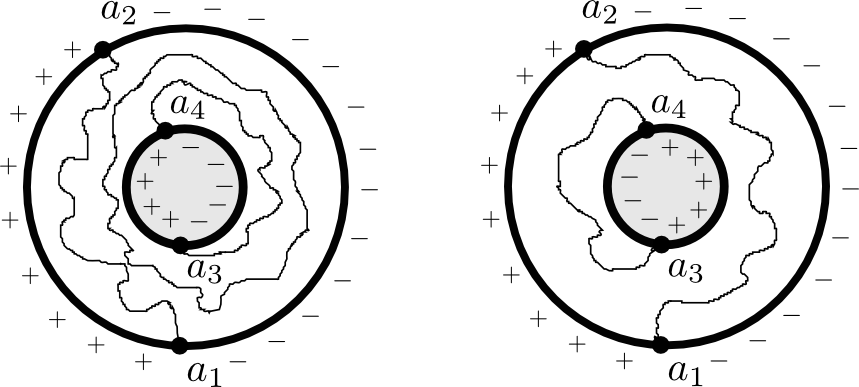}
\end{center}
\caption{\label{fig: spin} The fixed spin boundary conditions of Proposition \ref{prop: spin}. Both configurations contribute to $Z^\delta$; however, they correspond to $\tilde{Z}^\delta$ with boundary spins fixed in two different ways.}
\end{figure}

Let $Z^\delta$ and $\tilde{Z}^\delta$ denote the partition functions of the critical Ising model on the same domain $\Od$ with the same marked boundary edges $a^\delta_1,\dots,a^\delta_{2k}$, with locally monochromatic and fixed (in one of the possible ways) boundary conditions, respectively. Let us assume the following result to be true: the limit
$$
\Gamma(\Omega,a_1,\dots,a_{2k}):=\lim\limits_{\delta\to 0}\frac{\tilde{Z}^\delta}{{Z}^\delta}
$$
exists, is conformally invariant, depends Carath\'eodory continuously on $\Omega$ and smoothly on the positions of marked points, and all its derivatives are Lipschitz under $C_1$ perturbations of domains. We define 
$$
\tilde{Z}(\Omega,a_1,\ldots,a_{2k}):=Z(\Omega,a_1,\ldots,a_{2k})\Gamma(\Omega,a_1,\ldots,a_{2k}),
$$
where $Z$ is the SLE${}_3$ partition function defined in Section \ref{sec: intro_pf}. Then, we have the following proposition: 

\begin{proposition}
\label{prop: spin}
Conditionally on the above assumption, the driving process $a^\delta_1(t)$, defined as in Theorem \ref{thm: conv_int}, but this time for the spin Ising model with fixed boundary conditions, locally converges to the SLE${}_3$ with partition function $\tilde{Z}$.
\end{proposition}

\subsection{Fermionic observables and the SLE$_{3}$ partition functions.}
\label{sec: intro_pf}
The SLE${}_3$ partition functions involved in the statements of the above results are defined in terms of the \emph{continuous spinor fermionic observables}. Let $\Cvr:\widetilde{\Omega}\to\Omega$ be a double cover. We say that a function $f:\widetilde{\Omega}\to\C$ is a $\Cvr$-\emph{spinor} if it satisfies $f(z)=-f(z^*)$ for fibers $z,z^*$ over the same point of $\Omega$. Given a domain $\Omega$ such that $\pa \Omega$ consist of finitely many smooth Jordan curves, equipped with a double cover $\Cvr$ and a point $a\in\pa\Omega$, we define the observable $f_\Cvr(\Omega,a,\cdot)$ to be the unique holomorphic $\Cvr$-spinor that is continuous up to $\pa\Omega\setminus\{a\}$ and solves the following Riemann boundary value problem: 
\begin{gather}
f_\Cvr(\Omega,a,z)\sqrt{i n_z}\in \R,\quad z\in \pa \Omega \setminus \{a\}, \label{eq: obs_on_bdry}\\
f_\Cvr(\Omega,a,z)= \frac{\sqrt{i n_a}}{z-a}+O(1),\quad z\to a, \label{eq: obs_at_a}
\end{gather}
where $n_z$ stands for the outer normal at the point $z$. It was shown in \cite[Section~3]{ChIz} that the solution to this boundary value problem exists and is unique. Note that if $\varphi:\Omega\to\Omega'$ is a conformal map, then 
\begin{equation}
f_\Cvr(\Omega,a,z)=|\varphi'(a)|^{\frac12}\varphi'(z)^{\frac12}f_{\Cvr'}\left(\Omega',\varphi(a),\varphi(z)\right),
\label{eq: conf_inv_f}
\end{equation}
since the right-hand side solves the above boundary value problem in $\Omega$. We use this formula to \emph{define} the observable $f_\Cvr(\Omega,a,z)$ whenever $\Omega$ is unbounded or has a non-smooth boundary away from $a$, in particular, if $\Omega=\Lambda$ is an almost-$\H$ domain. In the latter case, one has
\begin{equation}
f_{\Cvr}(\Lambda,a,z)=O\left(|z|^{-1}\right),\quad z\to \infty. 
\label{eq: f_at_infty}
\end{equation}
Alternatively, one could define $f_{\Cvr}(\Lambda,a,z)$ as the unique solution to the boundary value problem (\ref{eq: obs_on_bdry})--(\ref{eq: obs_at_a})  satisfying (\ref{eq: f_at_infty}). If the boundary of $\Omega$ is not smooth at $a$, then $f_\Cvr(\Omega,a,z)$ is well-defined (via (\ref{eq: conf_inv_f}) with $|\varphi'(a)|^{\frac12}$ dropped) up to a real non-zero factor. 

For a configuration $(\Omega,a_1,\ldots,a_{2k})$, define $\Cvr_Z$ to be the double cover that branches around the boundary components containing an \emph{odd} number of marked points $a_1,\ldots,a_{2k}$. The SLE${}_3$ partition function is defined by the following Pfaffian formula:
\begin{equation}
Z(\Omega,a_1,\dots,a_{2k}):=\mathrm{Pf} \left[\sqrt{in_{a_r}}f_{\Cvr_Z}(\Omega,a_m,a_r)\right]_{1\leq m,r\leq 2k}.
\label{eq: defZ}
\end{equation}
It follows from (\ref{eq: conf_inv_f}) and the homogeneity property of the Pfaffian that under conformal maps one has 
\begin{equation}
Z(\Omega,a_1,\dots,a_{2k})=\prod_{i=1}^{2k}|\varphi'(a_i)|^{\frac12}\cdot Z\left(\varphi(\Omega),\varphi(a_1),\dots,\varphi(a_{2k})\right).
\label{eq: conf_cov_Z}
\end{equation}

In the SLE literature, a multiple SLE${}_\kappa$ partition functions is defined as a quantity attached to each (simply or finitely connected) domain with a smooth boundary and obeying the conformal covariance property (\ref{eq: conf_cov_Z}), the exponent $\frac12$ generally being $\frac{6-\kappa}{2\kappa}$. It is then postulated that 
\begin{equation}
\frac{Z(\Omega\setminus \gamma_{[0,t]},\gamma_t,a_2,\ldots,a_{2k})}{Z(\Omega'\setminus \gamma_{[0,t]},\gamma_t,a_{2k+1})}
\label{eq: part_function_SLE} 
\end{equation}
is the Radon-Nikodim derivative of the law of the curve $\gamma_t$ in the multiple SLE with respect to the chordal SLE in the \emph{simply-connected} domain $\Omega'$ (obtained from $\Omega$ by filling the holes) from $a_1$ to an auxiliary point $a_{2k+1}$. Note that due to the cancellation of the factors $|\varphi'(a_1)|^{\frac{6-\kappa}{2\kappa}}$, this ratio is well defined even though the boundary is not smooth at~$\gamma_t$. For the above scheme to be consistent, the ratio (\ref{eq: part_function_SLE}) must be a chordal SLE martingale. We warn the reader that this in not necessarily the case in our notation; one might need an additional factor independent of $a_1$. Such a factor does not change the logarithmic derivative and hence it is irrelevant for the statements of our results; therefore, we prefer not to keep track of it.

Our analysis bypasses a direct comparison with the chordal SLE${}_3$. Instead, the SLE${}_3$ partition functions appear as normalizing factors for the continuous (multi-point) martingale observables, which in turn are scaling limits of discrete fermionic observables. They are directly related to the underlying Ising model, namely, the ratios
$$
\frac{Z(\Omega,a_1,a_2,\dots,a_{2k})}{Z(\Omega,a'_1,a'_2\dots,a'_{2k})}
$$
are in fact limits of the ratios of the Ising model partition functions with corresponding boundary conditions (under the regularity assumption on the boundaries of the approximations $\Od$ to $\Omega$ at $a_1,a_1'\ldots,a_{2k},a'_{2k}$, see \cite{ChIz} for further details). 

Let us specialize the formula (\ref{eq: defZ}) to a few particular cases. In the simply-connected case, the only double cover is trivial, and it is easy to see that the function $\frac{1}{z-a}$ satisfies all the conditions of $f(\H,a,z)$ above. Therefore, one has explicitly
\begin{equation}
Z(\H,a_1,\dots,a_{2k})=\mathrm{Pf}\left[(a_m-a_r)^{-1}\right].\label{eq: Z_simply_conn}
\end{equation}

In the case of an annulus $\A_p=\{e^{-p}\leq |z|\leq 1\}$, $0<p<\infty$, it is convenient to write the partition function in terms of the arguments $\theta_{1},\ldots,\theta_{2k}$ of $a_1,\ldots,a_{2k}\in\pa\A_p$. The partition function for any doubly-connected smooth domain $\Omega$ can then be obtained by (\ref{eq: conf_cov_Z}). Denote  $\zeta_{mr}=0$ if $a_m$ and $a_r$ belong to the same boundary component, and $\zeta_{mr}=1$ otherwise. Writing explicitly the solution to the system (\ref{eq: obs_on_bdry})--(\ref{eq: obs_at_a}) in the annulus $\A_p$ (see Section \ref{DP_explicit}), we obtain 
\begin{eqnarray}
 Z(\A_p,\theta_1,\dots,\theta_{2k})=\mathrm{Pf}\left[Z^{\zeta_{rm}}_\text{odd}(p,\theta_m,\theta_r)\right],& \text{or}  \label{eq: pf_ann_1}\\
 Z(\A_p,\theta_1,\dots,\theta_{2k})=\mathrm{Pf}\left[Z^{\zeta_{rm}}_\text{even}(p,\theta_m,\theta_r)\right],\label{eq: pf_ann_2}
\end{eqnarray}
depending on whether the number of marked points on each boundary component is odd or even. Here
\begin{eqnarray*}
 Z^{0}_\text{odd}\left(p,\theta_m,\theta_r\right)=cs\left(\theta_r-\theta_m|\pi,p\right);& Z^{1}_\text{odd}(p,\theta_m,\theta_r)=dn(\theta_r-\theta_m|\pi,p);\\
 Z^{0}_\text{even}(p,\theta_m,\theta_r)=ds(\theta_r-\theta_m|\pi,p);& Z^{1}_\text{even}(p,\theta_m,\theta_r)=cn(\theta_r-\theta_m|\pi,p).
\end{eqnarray*}
The functions in the right-hand side are the Jacobian elliptic functions \cite{Spec_functions} with quarter-periods $K=\pi$ and $iK'=ip$. In the classical notation, $cs(z|\pi,p)=cs(\alpha(p)z;k(p))$ with the scaling factor $\alpha(p)$ and the modulus $k(p)$ given in terms of Jacobian theta constants as $\alpha(p)=\frac12\theta_2^2(0,e^{-p})$ and $k(p)=\theta_2^2(0,e^{-p})\theta_3^{-2}(0,e^{-p})$, and similarly for $dn,cn,ds$.
Finally, we give a formula for the partition function for the case of the annulus $\A^p$ with a marked point $a=e^{i\theta_a}$ on the boundary and $b=e^{-\rho_b+i\theta_b}$ inside:
\begin{equation}
Z_e(p,\theta_a,\theta_b,\rho_b)= \left(\frac{-i\wp'(\pi+(\rho_b-p) i| 2\pi; 2pi)}{\wp(\theta_b-\theta_a-\pi-ip| 2\pi, 2pi)-\wp(\pi+(\rho_b-p) i | 2\pi, 2pi)}\right)^{\frac12},
\label{eq: pf_annulus_radial}
\end{equation}
where $\wp$ is the Weierstrass elliptic $\wp$-function. See Section \ref{DP_explicit} for the proof.
 
\subsection{Remarks on the proof and related results.}

The main ingredient of our proof is a computation of the scaling limit of the martingale observable carried out in \cite{ChIz}. This observable generalizes one in \cite{CHS3}; we stress that in general, to ensure the martingale property, it is necessary to use the spinor version introduced in \cite{ChIz}. The rest of the proof closely follows one implemented by D. Zhan \cite{Zhan-LERW} for the convergence loop-erased random walks in multiply-connected domains. It is fair to say that since that work (and the earlier work regarding simply-connected domains \cite{LSW}) it is known that the convergence of holomorphic or harmonic martingale observable implies the convergence of driving forces of interfaces. In the most general situation of multiply-connected domains, explicit expressions for the observable are usually unavailable or hard to deal with;  however, to identify the scaling limits of driving forces, it suffices to know the conformal covariance properties of the observable and its behavior near the 
marked 
point $a_1$.  The proof is rather technical and involves a lot of details, mostly regarding the regularity properties of the observables and variational formulae. An important trick, also 
employed in \cite{Zhan-LERW}, is the use of compactness arguments with respect to Carath\'eodory topology in order to get uniform estimates from pointwise ones. 

The explicit description of the scaling limits of multiple Ising interfaces in simply-connected domains with four boundary points obtained in the present paper was conjectured in \cite{BBK05}. Standard SLE techniques together with a priori Ising model estimates allow one to prove the formula for the scaling limit of relative probabilities of two possible types of connection, also settling the conjecture from \cite[Section 8.3]{BBK05}, see further details in~\cite{Disser}. In the case of doubly-connected domains with two marked points on the boundary, the corresponding driving processes first appeared in \cite{Zhan-WP} as examples of reversible annulus SLE${}_3$. 

\subsection{Acknowledgements}
The author is grateful to the referee for valuable remarks. Work supported by Academy of Finland.

\section{Spinor observables and martingale property}
\setcounter{equation}{0}

\label{Sec: Obs_sc_Limits}

\subsection{Observables and their convergence}

\label{Sec: notations}
In this section we introduce the notation and recall the definitions and results from \cite{ChIz} concerning convergence of fermionic observables. By a \emph{discrete domain} $\Od$ we mean a connected finite subset of faces of $\mathbb{Z}^2$, together with all their incident vertices and edges. By \emph{boundary edges} of $\Od$ we mean ones that are incident to a vertex of $\Od$, but not to a face in $\Od$. If such an edge is incident to two different vertices in $\Od$, we naturally consider it as two distinct edges. Edges incident to a face in $\Od$ are called inner.

We will work with the low-temperature contour representation of the critical Ising model in
$\Od$ (see~\cite{Palmer07}). We denote by $\Conf(\Od)$ the set of all subsets $S$ of its inner edges, such that each vertex in $\Od$ is incident to $0,2$ or $4$ edges $S$. Further, given a collection $a_1^\delta,\ldots,a_{2k}^\delta$ of boundary edges of $\Od$, we connect them by $k$ arbitrary paths $\gamma_1,\dots,\gamma_k$ running along the edges of $\Od$, and define  
$$
\Conf(\Od,a_1^\delta,\ldots,a^\delta_{2k}) = \{S\bigtriangleup\gamma_1\bigtriangleup\dots\bigtriangleup\gamma_k:S\in\Conf(\Od)\},
$$
where $\bigtriangleup$ stands for the symmetric difference. It is easy to see that this set does not depend on the choice of $\gamma_1,\ldots,\gamma_k$, and that every $S\in \Conf(\Od,a_1^\delta,\ldots,a^\delta_{2k})$ can be decomposed into a collection of loops and $k$ paths connecting the marked points, such that none of those have transversal intersections or self-intersections, see Figure \ref{fig: Conf}. We will use the same notation when $a_1^\delta,\ldots,a^\delta_{2k}$ are edges on a double cover of $\Od$, identifying them with cossesponding edges in the base.

The \emph{spinor observable} (see \cite{ChIz}) is a function of a boundary edge $a^\delta$ and an inner or boundary edge $z$ of the double cover $\Cvr:\widetilde{\Omega}^\delta\to\Od$, defined by
\begin{equation}
\label{Observable} F_\Cvr(\Od,a^\delta,z):=i(in_{a_1^\delta})^{-\frac{1}{2}}\sum_{S\in \Conf(\Od,a^\delta,z)} W_\Cvr(a^\delta,z,S) x^{|S|}.
\end{equation}
where $W_\Cvr(a^\delta,z,S):= e^{-\frac{i}{2}\wind(\gamma)}\cdot (-1)^{l(S)+\mathbf{1}_{\gamma_{a^\delta\to z}}}$, and 
\begin{itemize}
 \item $\wind(\gamma)$ is the winding of the curve $\gamma\subset S$ connecting $a^\delta$ to $z$;
 \item $l(z)$ is the number of ``non-trivial'' loops in $S$ (i. e. those that do not lift as loops to the double cover $\widetilde{\Omega}^\delta$);
 \item $\mathbf{1}_{\gamma_{a^\delta\to z}}$ is the indicator of the event that $\gamma$ lifts as a path from $a^\delta$ to $z$ on the double cover;
 \item the outer normal $n_{a^\delta}$ is defined to be the edge $a^\delta$, viewed as a complex number, oriented outwards and normalized to have absolute value 1. 
\end{itemize}
For a domain with $2k-1$ marked points $a^\delta_1,\dots,a^\delta_{2k-1}$ on the boundary, the observable is defined by a similar formula 
\begin{equation}
\label{eq: Obs_many} F_\Cvr(\Od,a^\delta_1,\dots,a^\delta_{2k-1},z):=i(in_{a_1^\delta})^{-\frac{1}{2}}\sum_{S\in \Conf(\Od,a_1^\delta,\ldots,z)} W_\Cvr(a^\delta_1,\ldots,z,S) x^{|S|},
\end{equation}
where the winding for the curve $\gamma$ from $a_1$ to $z$ is defined by adding ``artificial arcs'' connecting $a^\delta_2$ to $a^\delta_3$, $a^\delta_4$ to $a^\delta_5$ etc., so that $\gamma$ really becomes a curve from $a^\delta_1$ to $z$  (see \cite[Section 5]{ChIz}). An alternative representation, essentially due to Hongler \cite{HThesis},  expresses this observable in terms of the basic one \cite[Proposition 5.8]{ChIz}: 
\begin{equation}
\label{eq: obs_pfaff}
 F_\Cvr(\Od,a^\delta_1,\dots,a^\delta_{2k-1},z)= \pm\sum\limits_{s=1}^{2k-1} (-1)^{s}\,\mathrm{Pf}\,P_\Cvr[s]\cdot F_\Cvr(\Od,a^\delta_{s},z).
\end{equation}
Here $\mathrm{Pf}\,P_\Cvr[s]$ denotes the Pfaffian of the real antisymmetric matrix   
\begin{equation}
\label{eq: def_G}
(P_\Cvr)_{m,r}:= \frac{F_\Cvr(\Od,a_m,a_{r})}{iF_\Cvr(\Od,a_r,a_r)} = - \frac{F_\Cvr(\Od,a_r,a_{m})}{iF_\Cvr(\Od,a_m,a_m)}=:-(P_\Cvr)_{r,m},\quad 1\leq m,r\leq 2k-1,
\end{equation}
from which $s$-th row and column are removed.

A crucial role in our analysis is played by the observables (\ref{eq: Obs_many}) with the double cover $\Cvr=\Cvr_Z$ defined in Section \ref{sec: intro_pf}. It was shown in \cite[Proposition 5.6]{ChIz} that 
\begin{equation}
\label{eq: obs_Z}
 F_{\Cvr_Z}(\Od,a^\delta_1,\dots,a^\delta_{2k-1},a^\delta_{2k})= (in_{a^\delta_{2k}})^{-\frac{1}{2}} Z(\Od,a^\delta_1,\dots,a^\delta_{2k}).
\end{equation} 
This equality will be essential for the martingale property of the observable. For the rest of the paper, we set
$$
F:=F_{\Cvr_Z}
$$
and will usually omit the double cover from the notation.

The main result in \cite{ChIz} is the following theorem relating the discrete fermionic observables to the continuous ones introduced in Section \ref{sec: intro_pf}.
\begin{theorem}{\cite[Theorem 3.13 and Lemma 4.8]{ChIz}} 
\label{thm: mainconv_int_1}
 Suppose that $\Od\stackrel{\text{Cara}}{\longrightarrow}\Omega$ as $\delta\to 0$, and that $a^\mesh\in \partial \Od$ tends to $a\in\pa \Omega$. Then, there exists a sequence of normalizing factors $\beta(\delta)=\beta(\delta,\Od,a^\delta,\Cvr)$ such that
$$
\beta(\delta) F_\Cvr(\Od,a^\delta,\cdot)\to f_\Cvr(\Omega,a,\cdot),\quad \delta\to 0,
$$
uniformly on compact subsets of $\Omega$. Moreover if $b^\delta\in\partial\Od$ tends to $b\in \Od$ and the convergence of domains is regular at $b$ (see \cite[Definition 3.14]{ChIz}), then $\beta(\delta)F_\Cvr(\Od,a^\delta,b^\delta)\to f_\Cvr(\Omega,a,b)$.
\end{theorem}

The \emph{continuous multi-point observable} is defined as the linear combination
\begin{equation}
f_\Cvr(\Lambda,a_1,\dots,a_{2k-1},z):=\sum\limits_{s=1}^{2k-1} P_s f_\Cvr(\Lambda,a_s,z)\label{eq: obs_tilda}
\end{equation}
with the real coefficients $P_s=P^\Cvr_{s}(\Lambda,a_1,\ldots,\hat{a}_s,\ldots,a_{2k-1})$ given by the Pfaffians
$$
P_s:=(-1)^{s+1} \mathrm{Pf} \left[\sqrt{in_{a_r}}f_\Cvr(\Lambda,a_m,a_r)\right]_{1\leq m,r\neq s\leq 2k-1}.
$$

We then have the following corollary: 
\begin{corollary}
\label{cor: conv_obs}
 If $\Od\stackrel{\text{Cara}}{\longrightarrow}\Omega$ regularly at $a_2,\dots,a_{2k}$ as $\delta\to 0$, then there exist normalizing factors $\beta(\delta)=\beta(\delta,\Od,\{a_j^\delta\},\Cvr)$ such that 
$$
\beta(\delta)F_\Cvr(\Od,a_1^\delta,\dots a^\delta_{2k-1},\cdot)\to f_\Cvr(\Omega,a_1,\dots,a_{2k-1},\cdot)
$$
uniformly on compact subsets of $\Omega$, and 
$$
\beta(\delta)F_\Cvr(\Od,a_1^\delta,\dots a^\delta_{2k-1},a^\delta_{2k})\to f_\Cvr(\Omega,a_1,\dots,a_{2k-1},a_{2k}).
$$
\end{corollary}
\begin{proof}
 We use the expansion \ref{eq: obs_pfaff}; observe that 
$$
(P_\Cvr)_{m,r}= \left(i\sqrt{in_{a^\delta_r}}F_\Cvr(\Od,a_r,a_r)\right)^{-1}i\sqrt{in_{a^\delta_r}}F_\Cvr(\Od,a_m,a_r)\;
$$
the real non-zero factor $(i\sqrt{in_{a^\delta_r}}F_\Cvr(\Od,a_r,a_r))^{-1}$ in fact does not depend on $r$. Thus, if one takes 
$$
\beta(\delta):=\prod\limits_{s=1}^{2k-1}\beta_s(\delta) \cdot \left(i\sqrt{in_{a^\delta_1}}F(\Od,a_1,a_1)\right)^{2k-2},
$$ 
where $\beta_s(\delta)$ is the factor from Theorem \ref{thm: mainconv_int_1} applied to $F_\Cvr(\Od,a_s,\cdot)$, then linearity of the Pfaffian guarantees the result.
\end{proof}

\subsection{Domain slit by the interface and the martingale property}
\label{sec: mart}
Here we invoke the martingale property of the observable. The computation is based on (\ref{eq: obs_Z}) and is essentially the same as \cite[Remark 2.4]{CHS2} and \cite[Proposition~9]{HonKyt}. For the sake of completeness, we include the proof.

Suppose that a domain $\Od$, the marked points $a^\delta_1,\ldots,a_{2k}^\delta\subset\partial\Od$ and a configuration $S\in \Conf(\Od, a^\delta_1,\dots,a^\delta_{2k})$ are given. Denote by $\gamma$ the curve in the decomposition of $S$ that starts at $a^\delta_1$, and by $\gamma_{[0,n]}$ its initial segment containing $n$ edges $a^\delta_1=\gamma_0,\dots,\gamma_n$. We assume any deterministic or random rule to resolve ambiguities arising in the decomposition of $S$. 

\begin{proposition}
\label{prop: mart}
 Let $\mathfrak{F}_n$ denote the filtration on the space of all configurations induced by the initial segments $\gamma_{[0,n]}$ of the curve $\gamma$. Then, for any edge $z\in \Od$, the process
\begin{equation}
\label{Mart}
M^\delta_{\gamma_{[0,n]}}(z):=
\frac{F_{\Cvr_Z}(\Od\setminus \gamma_{[0,n]},\gamma_n,a^\delta_2,\dots,a^\delta_{2k-1},z)}{F_{\Cvr_Z}(\Od\setminus \gamma_{[0,n]},\gamma_n,a^\delta_2,\dots,a^\delta_{2k})},
\end{equation}
stopped at the moment it becomes non-defined (that is, the interface hits or separates one of the boundary components different from $\partial_{a_1}$ or marked points), is a martingale with respect to $\mathfrak{F}_n$.
\end{proposition}
\begin{proof}
Denote by $Z_{\gamma_{[1,n]}}$ the restriction of the partition function to the set of configurations that have $\gamma_{[1,n]}$ as the initial segment of the interface. If $e$ is an edge adjacent to $\gamma_n$, denote $p_e:=\P\left[\left.\gamma_{[1,n+1]}=\gamma_{[1,n]}\cup e\right|\mathfrak{F}_n\right]$. Then.
$$
p_e=\frac{Z_{\gamma_{[1,n]}\cup\{e\}}}{Z_{\gamma_{[1,n]}}}=\frac{x^{n+1}Z(\Od\backslash(\gamma_{[1,n]}\cup \{e\}),e,\dots,a^\delta_{2k})}{x^{n}Z(\Od\backslash\gamma_{[1,n]},\gamma_n,\dots,a^\delta_{2k})}=\frac{x F(\Od\backslash(\gamma_{[1,n]}\cup e),e,\dots,a^\delta_{2k})}{F(\Od\backslash\gamma_{[1,n]},\gamma_n,\dots,a^\delta_{2k})},
$$
where we have used (\ref{eq: obs_Z}) in the second equality. Therefore,
\begin{multline*}
\E\left[ \left. \frac{F(\Od\backslash(\gamma_{[1,n+1]},\gamma_{n+1},\ldots,z)}{F(\Od\backslash(\gamma_{[0,n+1]},\gamma_{n+1},\ldots,a_2)}
\right|\mathfrak{F}_n\right] =\sum_e p_e \frac{F(\Od\backslash(\gamma_{[1,n]}\cup e),e,\dots,z)}{F(\Od\backslash(\gamma_{[1,n]}\cup e),e,\dots,a^\delta_{2k})}\\= 
\frac{\sum_e x F(\Od\backslash(\gamma_{[1,n]}\cup e),e,\dots,z)}{F(\Od\backslash\gamma_{[1,n]},\gamma_n,\dots,a^\delta_{2k})}=
\frac{F(\Od\backslash\gamma_{[1,n]},\gamma_n,\ldots,z)}{F(\Od\backslash\gamma_{[1,n]},\gamma_n,\dots,a^\delta_{2k})}.
\end{multline*} 
\end{proof}

Strictly speaking, $\Od\setminus\gamma_{[0,n]}$ is not a discrete domain in the sense formulated above, hence Theorem \ref{thm: mainconv_int_1} as given in \cite{ChIz} cannot be applied. To handle this issue, define a domain $\Od_n$ to be the face-connected component of $\Od\backslash \{\text{faces incident to }
\gamma_{[0,n]}\}$ that contains the marked point $a^\delta_2$; see Figure \ref{Fig: cut}. We will only consider the process $\gamma_{[0,n]}$ as long as $\Od_n$ also contains $a^\delta_3,\dots,a^\delta_{2k}$ and has the same topology as $\Od$. Let $v_n$ be the first boundary vertex of $\Od_n$ visited by a continuation $\gamma_n,\gamma_{n+1},\dots$ of the interface. Note that $v_n$ does not depend on the continuation. Indeed, if there were two such vertices $v_n,v'_n$, then one could consider a region enclosed by the boundary arc of $\Od_n$ between $v_n$ and $v_n'$ and the segments of the two continuations of $\gamma_{[0,n]}$ from their last common vertex to $v_n$ and $v_n'$ respectively. Faces in that region cannot be touched by $\gamma_{[0,n]}$, and thus they must belong to $\Od_n$, which is a contradiction. 

\begin{figure}[t]
\begin{center}
\includegraphics[width=0.7\textwidth]{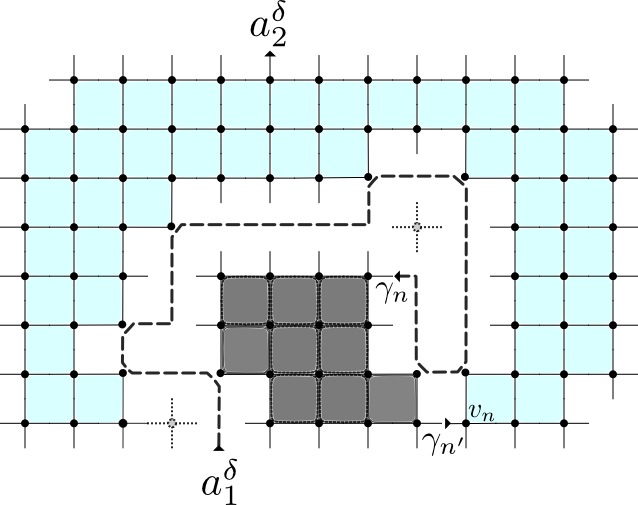}
\end{center}
\caption{\label{Fig: cut} A discrete domain $\Od$ slit by an initial segment $\gamma_{[0,n]}$ (dashed line) of the interface connecting $a^\delta_1$ to $a^\delta_2$. The domain $\Od\setminus \gamma_{[0,n]}$ is divided into $\Od_n$ (light blue squares) and $\overline{\Od_n}$ (bold gray squares). Any possible continuation of the interface enters $\Od_n$ throw the endpoint $v_n$ of the edge $\gamma_{n'}$. }
\end{figure}

Let $\overline {\Od_n}$ consist of all the edges of $\Od$ that belong neither to $\gamma_{[0,n]}$ nor to $\Od_n$. Let $\gamma_{n'}$ be the edge in one of the possible continuations $\gamma_n,\gamma_{n+1},\ldots$ traced right before visiting $v_n$; then $\gamma_{n'+1}$ is an edge of $\Od_n$. Applying one more time the above argument, one sees that if the interface $\gamma_{n'+1},\gamma_{n'+2},\ldots$ ever enters $\overline {\Od_n}$ again, it must return to $\Od_n$ through the same boundary vertex, and the same is true for loops and other interfaces of $S$ that have edges in $\Od_n$. Hence, changing the decomposition of $S$ if necessary, we can write it as $S=S'\cup S''$, where $S'$ and $S''$ are \emph{independent} configurations in $\Od_n$ and $\overline{\Od_n}$ respectively; more precisely,
$$
S'\cup\gamma_{n'+1}\in\Conf\left(\overline{\Od_n},\gamma_n,\gamma_{n'+1}\right)\quad \text{and}\quad S''\cup \gamma_{n'}\in\Conf(\Od, \gamma_{n'},a_2^\delta,\ldots a^\delta_{2k}).
$$
In particular,
\begin{equation}
\label{eq: decouple}
F_{\Cvr_Z}\left(\Od\setminus \gamma_{[0,n]},\gamma_n,\dots,z\right)=x^{-2}Z\left(\overline{\Od_n},\gamma_n,\gamma_{n'+1}\right)F_{\Cvr_Z}\left(\Od_n,\gamma_{n'},\dots,z\right).
 \end{equation}

\section{Proof of Theorem \ref{thm: conv_int}}

\subsection{A heuristic derivation of the driving process}
\label{sec: DP_general}

In this section, we give a heuristic derivation of the law of the limiting driving process $a_1(t)$ from the martingale property of the limit of $M^\delta_{\gamma_{[0,n]}}$. Suppose a random curve $\gamma_t\subset \Omega$ has a Loewner driving process $a_1(t)$, and the Loewner maps are denoted by $g_t(z)$, as defined before Theorem \ref{thm: conv_int}. Suppose we are given a family of analytic functions $M_{\gamma_{[0,t]}}(\cdot):\Omega\setminus \gamma_{[0,t]}\to\C$ with the following properties: first, for every $z\in\Omega$, $M_{\gamma_{[0,t]}}(z)$ is a local martingale, and second,
\begin{equation}
M_{\gamma_{[0,t]}}(z)=(g_t'(z))^\frac12\left(\frac{R_{\gamma_{[0,t]}}}{g_t(z)-a_1(t)}+Q_{\gamma_{[0,t]}}\left(g_t(z)\right)\right), 
\label{eq: mart_expansion}
\end{equation}
where $Q_{\gamma_{[0,t]}}(\cdot)$ extends analytically in a neighborhood of $a_1(t)$ and $Q_{\gamma_{[0,t]}}(a_1(t))=0$, and $R_{\gamma_{[0,t]}}$ does not depend on $z$. If $\gamma_t$ is a prospective scaling limit of the discrete Ising interface, then the scaling limit of $M^\delta_{\gamma_{[0,n]}}$, given by
\begin{equation}
M_{\gamma_{[0,t]}}(z)=\frac{f(\Omega\setminus \gamma_{[0,t]},\gamma_t,a_2,\ldots,a_{2k-1},z)}{f(\Omega\setminus  \gamma_{[0,t]},\gamma_t,a_2,\ldots,a_{2k-1},a_{2k})},                                                                                                                                                                             
\label{eq: def_M_cont}
\end{equation}
indeed satisfies this assumption. The expansion (\ref{eq: mart_expansion}), together with the analyticity $Q_{\gamma_{[0,t]}}(\cdot)$ at $a_1(t)$, follow directly from (\ref{eq: conf_inv_f}), (\ref{eq: obs_at_a}) and Schwarz reflection. We furthermore show in the beginning of Section \ref{sec: tech} that $O(1)$ in (\ref{eq: obs_at_a}) is in fact $o(1)$, and in the course of the proof of Lemma \ref{lem: non-vanish}, we prove the corresponding result for the multi-point observable; hence $Q_{\gamma_{[0,t]}}(a_1(t))=0$.

Assume that $M,R,Q$ depend nicely on $\gamma_{[0,t]}$, $R$ does not vanish, and $a_1(t)$ is a continuous semi-martingale. Then, the law of $a_1(t)$ is uniquely determined by the above conditions. Indeed, by It\^o's formula,
\begin{multline*}
 d\left(\frac{(g_t'(z))^\frac12R_{\gamma_{[0,t]}}}{g_t(z)-a_1(t)}\right)
 =
 \frac{R_{\gamma_{[0,t]}}d(g_t'(z))^\frac12+(g_t'(z))^\frac12(dR_{\gamma_{[0,t]}}+d[R,R]_t/2)}{g_t(z)-a_1(t)}\\
 -
 (g_t'(z))^\frac12R_{\gamma_{[0,t]}}\frac{dg_t(z)-da_1(t)-d[R;a_1]_t/R_{\gamma_{[0,t]}}}{(g_t(z)-a_1(t))^2}
 +
 \frac{(g_t'(z))^\frac12R_{\gamma_{[0,t]}}d[a_1]_t}{(g_t(z)-a_1(t))^3}.
\end{multline*}
Taking into account that $$d g_t(z)=\frac{2dt}{g_t(z)-a_1(t)}\quad \text{and}\quad d \log g'_t(z)=\frac{-2dt}{(g_t(z)-a_1(t))^{2}},$$ we obtain
\begin{multline*}
 d\left(\frac{(g_t'(z))^\frac12R_{\gamma_{[0,t]}}}{g_t(z)-a_1(t)}\right)
 =
 (g_t'(z))^\frac12R_{\gamma_{[0,t]}}\left(\frac{d[a_1]_t-3dt}{(g_t(z)-a_1(t))^3}\right.
 \\
 \left.+\frac{da_t+d[R;a_1]_t/R_{\gamma_{[0,t]}}}{(g_t(z)-a_1(t))^2}+\frac{(dR_{\gamma_{[0,t]}}+d[R,R]_t/2)/R_{\gamma_{[0,t]}}}{g_t(z)-a_1(t)}\right).
\end{multline*}
Similarly,  
\begin{multline*}
d\left((g_t'(z))^\frac12Q_{\gamma_{[0,t]}}\left(g_t(z)\right)\right)=\\(g_t'(z))^\frac12\left(-\frac{Q_{\gamma_{[0,t]}}(g_t(z))dt}{(g_t(z)-a_1(t))^2}+\frac{2Q'_{\gamma_{[0,t]}}\left(g_t(z)\right)dt}{g_t(z)-a_1(t)}+dQ_{\gamma_{[0,t]}}(w)|_{w=g_t(z)}\right).
\end{multline*}
Note that by our assumption on $Q_{\gamma_{[0,t]}}$, the first term in the brackets above is of order $(g_t(z)-a_1(t))^{-1}$ as $g_t(z)$ tends to $a_1(t)$. Since the drift term in $dM_{\gamma_{[0,t]}}(z)$ vanishes identically, so do the drift terms of the coefficients at $(g_t(z)-a_1(t))^{-3}$ and $(g_t(z)-a_1(t))^{-2}$. Thus we get
$$
 da_1(t)=\sqrt{3}dB_t-d[R;a_1]_t/R_{\gamma_{[0,t]}}.
$$
Using the definition (\ref{eq: Obs_many}), the conformal covariance property (\ref{eq: conf_inv_f}) and the expansion (\ref{eq: obs_at_a}), we can compute $R_{\gamma_{[0,t]}}$ explicitly:
$$
R_{\gamma_{[0,t]}}=\frac{P_1(\Lambda_t, a_2(t),\ldots,a_{2k-1}(t))}{g_t'(a_{2k})^{\frac12}f(\Lambda_t,a_1(t),\ldots,a_{2k}(t))}.
$$
By a well-known recursive formula for Pfaffians and the relation $n_{a_{2k}}g'_t(a_{2k})=n_{a_{2k}(t)}|g'_t(a_{2k})|$, we get
\begin{equation}
\label{eq: R}
R_{\gamma_{[0,t]}}=\frac{\sqrt{i n_{a_{2k}}}}{|g_t'(a_{2k})|^\frac12} \cdot\frac{\mathrm{Pf}\, [\sqrt{in_{a_r(t)}}f(\Lambda_t,a_m(t),a_r(t)]_{2\leq m,r\leq 2k-1}}{\mathrm{Pf}\, [\sqrt{in_{a_r(t)}}f(\Lambda_t,a_m(t),a_r(t))]_{1\leq m,r\leq 2k}}.
\end{equation}
Only the denominator of the second fraction depends on $a_1(t)$ (and hence non-smoothly on $t$). Therefore 
$$
d[R;a_1]_t/R_{\gamma_{[0,t]}}=-3\partial_{a_1}\log Z(\Lambda_t,a_1(t),\ldots, a_{2k}(t)),
$$
where $Z$ is as defined in (\ref{eq: defZ}).

Let us briefly indicate the steps needed to turn the above computation into a rigorous proof. First, we do not know \emph{a priori} that the scaling limit of the driving process exists and is a continuous semi-martingale. Although this can be derived from precompactness bounds as in \cite{CHS3}, we prefer another route and follow the argument of \cite{LSW}, dealing directly with $a^\delta_1(t)$. We consider a sequence of stopping times $t_1,t_2,\ldots$, with small increments $\Delta_{t_i}:=t_{i+1}-t_i$, and prove that the conditional expectations $\E[\Delta_{a_i}|\mathfrak{F}_{t_i}]$ and $\E[\Delta_{a_i}^2|\mathfrak{F}_{t_i}]$, where $\Delta_{a_i}:=a^\delta_1(t_{i+1})-a^\delta_1(t_{i})$, approximate the corresponding quantities 
for the diffusion $a_1(t)$ solving~(\ref{eq: a_t}). 

In order to compute the above conditional expectations, we use the martingale property of the discrete observable (\ref{Mart}). We express the increment of $M^\delta_{\gamma_{[0,n]}}$ between $n$ corresponding to the stopping times $t_i$ and $t_{i+1}$ in terms of $\Delta_{a_i}$ and $\Delta_{t_i}:=t_{i+1}-t_i$. To this end, we approximate the discrete observable by the continuous one (see Lemma \ref{lem: conv_ratio}), and then expand the latter in a Taylor series (see Lemma \ref{lem: variation}); in essence, this is the same computation as the one presented above.   By virtue of Lemma \ref{lem: coef_vanish}, we transform the martingale property into relations between $\E[\Delta_{a_i}^2|\mathfrak{F}_{t_i}]$, $\E[\Delta_{a_i}|\mathfrak{F}_{t_i}]$ and $\E[\Delta_{t_i}|\mathfrak{F}_{t_i}]$, and these relations allow one to apply a standard coupling argument to deduce convergence of $a_1^\delta$ to $a_1$. Along the way, we will prove various regularity estimates for the continuous observables. One particular 
problem to solve is a possible vanishing of $R_{\gamma_{[0,t]}}$ (which indeed may occur if $k>1$ and the domain is not simply connected). We will be able to solve it by renumbering the marked points; see Lemma~\ref{lem: generic}.

\subsection{Proof of Theorem \ref{thm: conv_int} modulo technical lemmas}
\label{sec: proof}
It will be assumed implicitly throughout the proof that $\delta\to 0$ along a sequence; of course, the convergence for \emph{any} such sequence also implies the same result when $\delta$ is a continuous parameter.

We fix a domain $(\Omega,a_1,\dots,a_{2k})$ and its approximations $(\Omega^\delta,a_1^\delta,\dots,a_{2k}^\delta)$,  maps $g_0$ and $g_0^\delta$ from $\Omega$ and $\Omega^\delta$ to almost-$\H$ domains $\Lambda_0$ and $\Lambda^\delta_0$, respectively. We require that $\partial\Lambda^\delta_0\to \partial\Lambda_0$ in the sense of $C^2$ (i. e. that there is a bijection between boundary components of $\partial\Lambda^\delta_0$ and $\partial\Lambda_0$, such that corresponding components can be parametrized by curves which are uniformly close and have uniformly close first and second derivatives).  We also require that $g_0^\delta\to g_0$ uniformly on compact subsets of $\Omega$; this implies that $a^\delta_m(0):=g^\delta_0(a^\delta_m)$ converges to $a_m(0):=g_0(a_m)$, $1\leq m\leq2k$. 

One way to construct such a sequence $g_0^\delta$ is as follows. Let $\Omega$ be $m$-connected, and let $\partial^\delta_1,\ldots,\partial^\delta_m$ denote connected components of $\C\setminus \Omega^\delta$,  $\partial^\delta_1$ being the unbounded one. Let $\varphi^\delta_1$ be conformal maps from $\C\setminus \partial^\delta_1$ to the unit disc $\D$ chosen so that $\varphi^\delta_1\stackrel{\delta\to 0}{\longrightarrow} \varphi_1$ as $\delta\to 0$.  Define the conformal maps $\varphi^\delta_2,\ldots,\varphi^\delta_m$ by $\varphi^\delta_{i+1}:=\psi^\delta_{i+1}\circ\varphi^\delta_i$, where $\psi^\delta_{i+1}$ is the conformal map from $\hat\C \setminus \varphi_i^\delta(\partial^\delta_{i+1})$ to $\hat\C \setminus \D$ satisfying $\psi^\delta_i(\infty)=\infty$ and $(\psi^\delta_i)'(\infty)>0$. By Carath\'eodory convergence of domains, we have that $\varphi^\delta_1(\hat\C\setminus\partial^\delta_2)\stackrel{Cara}{\longrightarrow}\varphi_1(\hat\C\setminus\partial_2)$, and therefore $\psi^\delta_2\to \psi_2$ 
and 
$\varphi^\delta_2\to \varphi_2$ uniformly on compact subsets of $\hat\C\setminus \varphi_1(\partial_2)$ and $\C\setminus (\partial_1\cup\partial_2)$ respectively. Iterating, we see that $\varphi^\delta_m\to\varphi_m$ uniformly on compact subsets of $\Omega$. Moreover, $\partial\varphi^\delta_m(\Od)$ converges to $\partial\varphi_m(\Omega)$ in the sense of $C^\infty$, since the image of the $i$-th boundary component under $\varphi^\delta_m$ coincides with the image of the unit circle under $\psi^\delta_m\circ\ldots\circ \psi^\delta_{i+1}$, and the uniform convergence of conformal maps implies the $C^\infty$ convergence. The maps $g_0^\delta$ are then constructed by taking composition with a conformal map from $\varphi_m(\C\setminus \partial_{a_1}^\delta)$ to the upper half-plane, where $\partial^\delta_{a_1}$ is the component containing $a^\delta_1$.

We fix a cross-cut $\crs$ in $\Lambda_0$ that separates $a_1(0)\in\R$ from other marked points and other boundary components. The constants in the proof below may depend on all these data (but not on anything else). Denote  $\crs_\Omega:=g_0^{-1}(\crs)$ and $\crs_{\Od}:=(g^\delta_0)^{-1}(\crs)$. Recall that $T^\delta_\crs$ is the hitting time of $\crs_{\Od}$ by the interface $\gamma_t$, where the latter is parametrized by twice half-plane capacity $t$ of its image $g_0^\delta(\gamma_t)$.

We denote by $\Theta_\Omega(\crs)$ the set of all possible pairs $(\Omega^\delta\setminus \gamma_{[0,t]},\gamma_t)$, viewed as a domain and a point on its boundary, where $t<T^\delta_\crs$. We denote by $\CSpace(\crs)$ the set of all possible tuples $(\Lambda^\delta_t,a^\delta_1(t),\dots,a_{2k}^\delta(t))$, $t<T^\delta_\crs$. Note that the closure of  $\Theta_\Omega(\crs)$ with respect to Carath\'eodory convergence of domains and the convergence of prime ends is compact, since the domains in $\Theta_\Omega(\crs)$ are uniformly bounded and all contain a fixed ball around a point $w$ separated from $a_1$ by $\crs_\Omega$; see \cite{Pommerenke}. Therefore, the closure of $\Theta(\crs)$, taken with respect to the topology induced from $C^2$-convergence of boundaries of almost-$\H$ domains and the usual convergence of marked points, compact and consists of domains of the same topology as $\Lambda_0$. Indeed, a convergent subsequence in $\Theta_\Omega(\crs)$ corresponds to a convergent one in $\Theta(\crs)$.

We now state several technical lemmas. The proofs are mainly based on the properties of the boundary value problem (\ref{eq: obs_on_bdry}) -- (\ref{eq: obs_at_a}) and compactness arguments, and are postponed to the next subsection.

\begin{lemma}
 \label{lem: non-vanish}
If $\partial \Omega$ is smooth near marked points $a_2,\dots,a_{2k}$, then $$f_{\Cvr_Z}(\Omega,a_1,\dots a_{2k-1},a_{2k})\neq 0.$$ 
\end{lemma}
In particular, this allows one to define the continuous analog of the martingale $M^\delta_{\gamma_{[0,n]}}$. Put 
$$
M_{\gamma_{[0,t]}}(z):=\frac{f_{\Cvr_Z}(\Omega^{\delta}\backslash \gamma_{[0,t]},\gamma_t,a^\delta_2,\dots,a^\delta_{2k-1},z^{\delta})}{f_{\Cvr_Z}(\Omega^{\delta}\backslash \gamma_{[0,t]},\gamma_t,a_2^\delta,\dots,a_{2k-1}^\delta,a^{\delta}_{2k})};
$$
here we slightly abuse the notation and write $a^{\delta}_{2},\ldots,a^{\delta}_{2k}$ for the endpoints of the corresponding boundary edges on the boundary of $\widetilde{\Omega}^\delta$; and $\Od$ is viewed as a continous (polygonal) domain. We also write  $\gamma_{[0,t]}$ for the parametrization of $\gamma\subset \Od$ by twice the half-plane capacity of $g_0^\delta(\gamma)$, whereas $\gamma_{[0,n]}$ is the parametrization by number of steps.

The following lemma is easily derived from Corollary \ref{cor: conv_obs} and the compactness (see Section~\ref{sec: tech}). In the next few lemmas, $\mathcal{C}\subset \Omega$ denotes an arbitrary compact set with non-zero interior separated from $a_1$ by the cross-cut $\crs_{\Omega}$.
\begin{lemma}
 \label{lem: conv_ratio}
Under conditions of Theorem \ref{thm: conv_int},
\begin{equation}
 \left|M^\delta_{\gamma_{[0,n_t]}}(z)-M_{\gamma_{[0,t]}}(z)\right|\stackrel{\delta\to 0}{\longrightarrow} 0 
\label{eq: conv_ratio}
\end{equation}
uniformly in $z\in\mathcal{C}\cap\mathrm{Edges}(\Od)$ and all possible initial segments $\gamma_{[0,t]}$, $t<T^\delta_\crs$, where $n_t$ stands for the smallest $n$ such that $2\mathrm{hcap}(g_0^\delta(\gamma_{[0,n]}))\geq t$.
\end{lemma}
To express the variation of $M_{\gamma_{[0,t]}}$ as $t$ grows, we will need spinors introduced in the following Lemma.
\begin{lemma}
\label{lem: f_n}
 Given an almost-$\H$ domain $\Lambda$, a point $a\in \R$, and an integer $n>0$, there is a unique $\Cvr_Z$-spinor $f^{(n)}(\Lambda,a,\cdot)$ satisfying (\ref{eq: obs_on_bdry}), (\ref{eq: f_at_infty}) and the expansion
 $$
 f^{(n)}(\Lambda,a,\cdot)=\frac{1}{(z-a)^{n+1}}+O(1), \quad z\to a.
 $$
 Moreover, $f^{(n)}(\Lambda,a,\cdot)\equiv \frac1n \partial_a f^{(n-1)}(\Lambda,a,\cdot)$.
\end{lemma}

Recall that $R=R_{\gamma_{[0,t]}}$ is defined by formula (\ref{eq: R}). Define $R'_{\gamma_{[0,t]}}$ to be the derivative of that expression with respect to $a_1$ (evaluated at $\Lambda=\Lambda^\delta_t$, $a_1=a^\delta_1(t)$,\dots, $a_{2k}=a^\delta_{2k}(t)$). 

Fix $\epsilon>0$. Let $t_{1}<t_2$ and a realization of $\gamma_{[0,t_2]}$ be such $t_2<T^\delta_\crs$. Assume that 
\begin{gather}
t_2-t_1\leq 2\epsilon^2,\label{eq: cond_t_i}\\                                                                                                                                                                                                                                                                
\max_{t\in[t_1,t_2]}|a^\delta_1(t)-a^\delta_1(t_1)|\leq 2\epsilon. \label{eq: cond_a_t_i}                                                                                                                                                                                                                                                                                    \end{gather}
Denote $\Delta_a:=a^\delta_1(t_2)-a^\delta_1(t_1)$ and  $\Delta_t:=t_2-t_1$. Then the following variational formula holds true:
\begin{lemma}
\label{lem: variation}
 Under the above conditions, on has 
\begin{align*}
M_{\gamma_{[0,t_2]}}(z)-M_{\gamma_{[0,t_1]}}(z)=(g^\delta_{t_1})'(z)^\frac12&\left[R\left(\Delta_a^2-3\Delta_t\right)f^{(2)}\left(\Lambda^\delta_{t_1},a^\delta_1(t_1),g^\delta_{t_1}(z)\right)\right.\\&+\left(R\Delta_a+R'\Delta^2_a\right)f^{(1)}(\Lambda^\delta_{t_1},a^\delta_1(t_1),g^\delta_{t_1}(z)) \\&+\left.\sum_{s=1}^{2k-1}c_sf(\Lambda^\delta_{t_1},a^\delta_s(t_1),g^\delta_{t_1}(z)) \right]+O(\epsilon^3),                                     
\end{align*}
where $R=R_{\gamma_{[0,t_1]}}$, $R'=R'_{\gamma_{[0,t_1]}}$, the constant in $O(\epsilon^3)$ is uniform over realizations of $\gamma_{[0,t_2]}$ and in $z\in\comp$; $c_s=c_s(\epsilon,\gamma_{[0,t_1]})\in \sqrt{in_{a_{2k}}}\R$.
\end{lemma}

This and the following lemma will allow one to use the martingale property of $M^\delta_{[0,n]}$ to extract information about conditional expectations of $\Delta_a$ and $\Delta_t$ given $\gamma_{[0,t]}$. 
\begin{lemma}
\label{lem: coef_vanish}
There is a constant $L=L(\comp)>0$ such that the following holds. Let $\gamma_{[1,t]}$ be a realization of the initial segment of the interface, $t<T^\delta_\crs$, and let $\Phi:\Lambda^\delta_{t}\to \C$ be a spinor of the form
\begin{equation}
\label{eq: Phi}
\Phi(z)=C_2f^{(2)}(\Lambda^\delta_{t},a_1(t),z)+ C_1f^{(1)}(\Lambda^\delta_t,a_1(t),z) + \sum\limits_{s=1}^{2k-1}c_s f(\Lambda^\delta_t,a_s(t),z) 
\end{equation}
with real constants $C_{1,2},c_1,\ldots,c_{2k-1}$. Then, 
$
|C_{1,2}|\leq L\max_{z\in g^\delta_t(\mathcal{C})}|\Phi(z)|.
$ 
\end{lemma}
There will be an additional technical problem caused by the fact that $R=R_{\gamma_{[0,t_1]}}$ might be small or vanish. It turns out that this situation can be ruled out by changing the order of points. Note that while the expression in (\ref{eq: R}) is symmetric under permutations of $a_2,\ldots,a_{2k-1}$, the points $a_1$ and $a_{2k}$ play a special role.
\begin{lemma}
\label{lem: generic}
There exists a constant $c>0$ such that for every possible realization of $\gamma_{[0,t]}$, $t<T^\delta_\crs$, there is a reordering of the points $a_2^\delta,\ldots,a^\delta_{2k}$ after which $|R_{\gamma_{[0,t]}}|>c$.
\end{lemma}

Recall that by definition, 
\begin{equation}
D(\Lambda,a_1,\ldots,a_{2k})=3\frac{\partial_{a_1}\mathrm{Pf}\,[\sqrt{in_{a_r}}f(\Lambda,a_m,a_r)]_{1\leq r,m\leq 2k}}{\mathrm{Pf}\,[\sqrt{in_{a_r}}f(\Lambda,a_m,a_r)]_{1\leq r,m\leq 2k}}.                                                       
\label{eq: def_D_bis}
\end{equation}
Therefore, $D(\Lambda^\delta_t,a^\delta_1(t),\ldots,a^\delta_{2k}(t))=-3R'_{\gamma_{[0,t]}}/R_{\gamma_{[0,t]}}$, as long as $R_{\gamma_{[0,t]}}\neq 0$. It will sometimes be convenient to view this quantity as a function of the values of the driving process $a^\delta_1(s)$, $s\leq t$. More generally, given an almost-$\H$ domain $\Lambda_0$ equipped with with marked points $a_1(0)\in\R$, \dots, $a_{2k}(0)$ and a cross-cut $\crs$ separating $a_1(0)$ from other marked points and other boundary components, write $D((a_1)_{[0,t]})$ for $D(\Lambda_t,a_1(t),\ldots,a_{2k}(t))$, where $a_1(t)$ and a continuous driving function of the Loewner chain $g_t$, $\Lambda_t=g_t(\Lambda)$ and $a_2(t)=g_t(a_2)$, \dots, $a_{2k}(t)=g_t(a_{2k})$. 

\begin{lemma}
\label{lem: D-lipshitz}Assume that $a_1(s)$ , $\tilde{a}_1(s)$ are two driving functions such that the corresponding hulls don't intersect $\crs$ at times $t$, $\tilde{t}$ respectively, and the initial conditions are given either by $\Lambda_0,a_1(0),\ldots,a_{2k}(0)$, or by  $\Lambda^\delta_0,a^\delta_1(0),\ldots,a^\delta_{2k}(0)$. Then 
$$
 \left|D((a_1)_{[0,t]})-D((\tilde{a}_1)_{[0,t]})\right|\leq L\left(|\tilde{t}-t|+\sup\limits_{0\leq s\leq t\wedge \tilde{t}} |a_1(s)-\tilde{a}_1(s)|\right),
$$
where a constant $L>0$ may depend on $\Lambda_0$, $\{a_i(0)\}$ and $\crs$, but not on $a_1(s)$, $\tilde{a}_1(s)$ or $\delta$.
\end{lemma}
The latter estimate, in particular, guarantees the existence and the uniqueness of a strong solution to (\ref{eq: a_t}) (in fact, $\sqrt{3} B_t$ may be replaced by any continuous function) up to $T_\crs$. This can be proven in the usual way by Picard-Lindel\"of iteration (see, e. g., \cite[Theorem 3.1]{Zhan-LERW}).
 
\begin{proof}[Proof of Theorem \ref{thm: conv_int}]
We follow the standard scheme proposed in \cite{LSW} and also employed e. g. in \cite{SS_harm, Zhan-LERW}. We introduce a small parameter $\epsilon$ (eventually, $\epsilon$ will be taken to zero, but $\delta$ will be much smaller than $\epsilon$), and \emph{stopping times} $0=t_0<t_1<t_2<\dots<t_{N(\crs)}=T^\delta_\crs$, so that $t_{i+1}$ is the smallest $t>t_i$ such that  either $t-t_i=\epsilon^2$, or $|a^\delta_1(t)-a^\delta_1(t_{i})|=\epsilon$, or $t=T^\delta_\crs$. Note that $t_i$ and $t_{i+1}$ satisfy (\ref{eq: cond_t_i}) and (\ref{eq: cond_a_t_i}).

Pick a compact set $\mathcal{C}\subset \Omega$ with a non-empty interior separated from $a_1$ by $\crs_\Omega$. By Lemma~\ref{lem: conv_ratio}, for every fixed $\epsilon>0$ there is a $\delta_0>0$ such that
$$
\left|M^\delta_{\gamma_{[0,n_{i}]}}(z)-M_{\gamma_{[0,t_{i}]}}(z)\right|\leq \epsilon^3
$$
for all $i<N(\crs)$, $z\in\comp$ and $\delta<\delta_0$; recall that $n_i=n_{t_i}$ stands for the smallest $n$ such that $2\mathrm{hcap}(g_0^\delta(\gamma_{[0,n]}))\geq t_i$. If $\mathfrak{F}_{t_i}$ denotes the sigma-algebra generated by $\gamma_{[0,t_i]}$, then $\mathfrak{F}_{t_i}=\mathfrak{F}_{n_i}$. Therefore, by Proposition \ref{prop: mart}, 
$$\left|
\E\left[M_{\gamma_{[0,t_{i+1}]}}(z)-M_{\gamma_{[0,t_{i}]}}(z)|\mathfrak{F}_{t_{i}}\right]\right|\leq 2\epsilon^3.
$$
for all $z\in\comp$ and $i<N(\crs)$. Denote $\Delta_{a_i}:=a^\delta_1(t_{i+1})-a^\delta_1(t_{i})$, $\Delta_{t_i}:=t_{i+1}-t_{i}$, $R_i:=R_{\gamma_{[0,t_i]}}$, $R'_i:=R'_{\gamma_{[0,t_i]}}$. Applying Lemma \ref{lem: variation}, we get that  $(g^\delta_{t_i})'(z)^\frac12\Phi_i(z)=O(\epsilon^3)$ uniformly in $z\in\comp$ and $i<N(\crs)$, where $\Phi_i(z)$ is of the form (\ref{eq: Phi}) with $t=t_i$ and 
\begin{eqnarray*}
C_2=\E[R_i\Delta_{a_i}^2-3R_i\Delta_{t_i}|\mathfrak{F}_{t_i}],\\
C_1=\E[R_i\Delta_{a_i}+R'_i\Delta^2_{a_i}|\mathfrak{F}_{t_i}].
\end{eqnarray*} 
Since $|(g^\delta_{t_i})'(z)|$ is uniformly bounded from below over $\comp$, Lemma \ref{lem: coef_vanish}, applied to $\Phi_i/\sqrt{in_{a_{2k}}}$ (which has real coefficients), implies 
\begin{eqnarray*}
\E[R_i\Delta_{a_i}^2-3R_i\Delta_{t_i}|\mathfrak{F}_{t_i}]=O(\epsilon^3),\\
\E[R_i\Delta_{a_i}+R'_i\Delta^2_{a_i}|\mathfrak{F}_{t_i}]=O(\epsilon^3),
\end{eqnarray*} 
uniformly in $i<N(\crs)$. By Lemma \ref{lem: generic}, for every $i$ we can reorder the marked points so that  $|R_i|>c$, hence 
\begin{eqnarray}
\E[\Delta_{a_i}^2-3\Delta_{t_i}|\mathfrak{F}_{t_i}]=O(\epsilon^3),\label{eq: exp_Delta_1}\\
\E[\Delta_{a_i}-D_i\Delta_{t_i}|\mathfrak{F}_{t_i}]=O(\epsilon^3),\label{eq: exp_Delta_2}
\end{eqnarray} where $D_i=D_{\gamma_{[0,t_i]}}$.
 Define a function $w(t)$ by 
 $ w(t_i):=a^{\delta}_1(t_i)-\sum\limits_{s=0}^{i-1}D_{s}\Delta_{t_s},$ extending it in a piecewise linear way between $t_i$. Then (\ref{eq: exp_Delta_1}) and (\ref{eq: exp_Delta_2}) imply
$$
\begin{aligned}
&\E[(w(t_{i+1})-w(t_i))^2-3\Delta_{t_i}|\mathfrak{F}_{n_i}]=O(\epsilon^3),\\
&\E[w(t_{i+1})-w(t_i)|\mathfrak{F}_{n_i}]=O(\epsilon^3).
\end{aligned}
$$
It is well-known (see \cite[Section 3.3]{LSW} that these relations allow one to construct a coupling of the process $w_t$ with a standard Brownian motion $B_t$ and a sequence of stopping times $\tau_i$ for $B_t$, such that $w(t_{i})=\sqrt{3}B_{\tau_i}$ and, moreover, the probability that $\max|\tau_i-t_i|> \sqrt{\epsilon}$ is $O(\epsilon)$. Note that by Lemma \ref{lem: D-lipshitz} and the continuity of Brownian motion, this implies that $\P[\max\limits_{t<\tau_{N(\crs)}}|w(t)-\sqrt{3}B_t|> \epsilon^{\frac15}]=O(\epsilon)$. Let us show that on this event $a^{\delta}_1(t)$ and the solution $a_1(t)$ to (\ref{eq: a_t}) are also uniformly close to each other with high probability, up to $t=T^\delta_{\crs'}\wedge T_{\crs'}$, where $\crs'$ is any cross-cut separating $g_0(a_1)$ from $\crs$.

 The bound of Lemma \ref{lem: D-lipshitz} ensures that the sum in the definition of $w(t)$ can be replaced by an integral with a small error, i. e. there is a constant $C>0$, such that 
$$\left|a^{\delta}_1(t)-w(t)-\int_0^tD\left((a^{\delta}_1)_{[0,s]}\right)ds\right|<C\epsilon
$$
We thus can write
$$
|a_1^{\delta}(t)-a_1(t)|\leq|w(t)-\sqrt{3}B_t|+\int_0^t\left|D\left((a^{\delta}_1)_{[0,s]}\right)-D\left((a_1)_{[0,s]}\right)\right|ds+C\epsilon.
$$  
By Lipschitzness of the drift function $D$ (see Lemma \ref{lem: D-lipshitz}), it follows that on the event that $|w(t)-\sqrt{3}B_t|<\epsilon':=\epsilon^{\frac15}$ one has
$$
|a_1^{\delta}(t)-a_1(t)|\leq\epsilon'+C\epsilon+L\int_0^t\max\limits_{[0,s]}|a_{1}^{\delta}(\cdot)-a_{1}(\cdot)|ds.
$$
From Gronwall's lemma, we conclude that $|a_1^{\delta}(t)-a_1(t)|\leq(\epsilon'+C\epsilon)e^{Lt}\leq(\epsilon'+C\epsilon)e^{2L\text{hcap}(\crs')}$ as long as $t<T^\delta_{\crs'}\wedge T_{\crs'}$. Taking $\epsilon$ to zero concludes the proof.
\end{proof}

\begin{proof}[Proof of Corollary \ref{cor: stronger topology}]
We assume that the reader is familiar with the proof in \cite{CHS3}, which goes by checking Condition C from \cite{AnttiStas}. The only additional problem in the present case is that \cite{AnttiStas} only deals with simply-connected domains. However, due to a local nature of our result, this is irrelevant. Let $\widehat{\Omega}^\delta$ denote the simply-conncted domain obtained from $\Od$ by filling the connected components of $\widehat{\C}\setminus \Od$, except for the one containing  $a_1^\delta$. Let $\tilde{\gamma}$ be an interface in $\widehat{\Omega}^\delta$ that consists of $\gamma_{[0,T^\delta_\crs]}$ and its continuation to some auxilliary point $\hat{a}^\delta\in \pa \hat{\Omega}^\delta$ separated from $a_1^\delta$ by $\crs_{\Od}$. This interface fits into the framework of \cite{AnttiStas}. Note that unless the hitting point $\gamma_{T^\delta_\crs}$ is close to the boundary of $\Od$, this continuation (and the point $\hat{a}_1$) can be chosen so that it does not create any additional unforced 
crossings of annuli of large modulus, and the initial segment $\gamma_{[0,T^\delta_\crs]}$ does such crossings with probability bounded away from $1$ by Corollary 1.7 in \cite{CDH}, as explained in \cite{CHS3}. Modifying $\crs_{\Od}$ if necessary and applying the same corollary once more, we see that the probability of $\gamma_{T^\delta_\crs}$ to be close to the boundary is small. Therefore, Corollary follows from the results of \cite{AnttiStas} applied to $\tilde{\gamma}$.
\end{proof}

\subsection{Proof of the technical lemmas}

\label{sec: tech}

We start by making several remarks concerning the boundary value problem of the type (\ref{eq: obs_on_bdry}) on almost-$\H$ domains. First of all, note that if $f$ satisfies (\ref{eq: obs_on_bdry}) on $\partial \Lambda\setminus \{a\}$, then Schwarz reflection allows one to extend $f$ to an analytic function on $\Lambda^*\setminus \{a\}$, where $\Lambda^*=\Lambda\cup\overline{\Lambda}\cup \R$. Furthermore, note that $h(z):=\Im\int^z f^2(z)dz$ is a harmonic function, locally constant on $\pa\Lambda\setminus\{a\}$ (and therefore single-valued), with non-negative inner normal derivative. This, in particular, implies that if $f$ satisfies (\ref{eq: obs_on_bdry}) \emph{everywhere} on the boundary, then $f\equiv 0$. We will always choose a constant of integration for $h$ in such a way that $h\equiv 0$ on $\R\setminus \{a\}$. (The condition (\ref{eq: f_at_infty}) ensures that the constants on $(-\infty;a)$ and $(a,\infty)$ coincide.) Therefore, if $a\in\R$ and $h$ corresponds to $f=f(\Lambda,a,\cdot)$, then 
$h(z)=-\Im\frac{1}{z-a}+O(\Im z)$ as $\quad z\to a$. Differentiating and taking the square root, we get
\begin{equation}
f(\Lambda,a,z)=\frac1{z-a}+Q(\Lambda,a,z) 
\label{eq: obs_expansion}
\end{equation}
with $Q(\Lambda,a,\cdot)$ holomorphic in $\Lambda^*$ and $Q(\Lambda,a,a)=0$. 

We will need several auxilliary lemmas. The following one will be a convenient tool to prove variational formulae for the observables. Denote $B_r(w):=\{z:|z-w|<r\}$.
\begin{lemma}
\label{lem: remove_singularity}
 Let $\Lambda_i$ be a sequence of almost-$\H$ domains converging to $\Lambda$ in the sense of $C^1$, $a\in\R$, and $r_i>0$ tends to $0$ as $i\to\infty$. Suppose $f_i:(\Lambda_i\cup\pa\Lambda_i)\setminus B_{r_i}(a)\to \C$ are continuous holomorphic spinors satisfying (\ref{eq: obs_on_bdry}) on $\pa\Lambda_i\setminus B_{r_i}(a)$ and (\ref{eq: f_at_infty}). Assume that for some $R>0$, one can write $f_i(z)\equiv u_i(z)+v_i(z)$, $z\in B_R(a)\setminus B_{r_i}(a)$, where $u_i(\cdot)$ is analytic in $B_R(a)$ and for every $0<r<R$, $v_i(\cdot)\to v(\cdot)$ uniformly on $B_R(a)\setminus B_r(a)$. Then, there is a unique holomorphic spinor $f:\Lambda\to\C$ satisfying (\ref{eq: obs_on_bdry}), (\ref{eq: f_at_infty}), and $f(z)=v(z)+O(1)$ as $z\to a$. Moreover, $f_i(z_i)\to f(z)$ whenever $z_i\to z\in(\Lambda\cup\partial\Lambda)\setminus\{a\}$.
\end{lemma}
\begin{proof}
Put $M_i=\sup_{\Lambda_i\setminus B_{R/2}(a)}|f_i|$ and assume first that $M_i$ are uniformly bounded. Then, for every $r>0$, $f_i$ are uniformly bounded on  $B_R(a)\setminus B_{r}(a)$, since $v_i$ converge uniformly on this set and 
$$
\sup_{B_R(a)\setminus B_{r} (a)}|u_i|\leq \sup_{B_R\setminus B_{R/2} (a)}|u_i|\leq M_i+\sup_{B_R\setminus B_{R/2} (a)}|v_i|,$$ 
where we have applied maximum principle to $u_i$. Therefore, we may assume (after extracting a subsequence) that $f_i$ converge to a spinor $f$ uniformly on compact subsets of $\Lambda^*\setminus \{a\}$. Applying once again the maximum principle to $u_i=f_i-v_i$, we infer that $u=f-v$ is bounded in a neighborhood of $a$. Consider the harmonic functions $h_i(w):=\Im \int^w f_i^2(z)dz$, and denote their constant values on the $m$ inner boundary components by $C_i^{(1)},\ldots,C_i^{(m)}$. From convergence of $f_i$ inside $\Lambda$ it is clear that these constants must converge to the boundary values $C^{(1)},\ldots,C^{(m)}$ of $h=\int f^2$, from which we in turn deduce that $f$ satisfies (\ref{eq: obs_on_bdry}) and (\ref{eq: f_at_infty}). Moreover, since $\pa \Lambda$ converges in the sense of $C^1$, it is easy to see that the convergence of $\nabla h_i$, and hence $f_i=\sqrt{2\partial_z h_i}$, holds ``up to the boundary'' as stated in the assertion. As explained above, the singularity at $a$ determines $f$ 
uniquely, hence the convergence holds without extraction.

Assume now that $M_i\to \infty$ along a subsequence. Then we can apply the above reasoning to $\tilde{f}_i:=M^{-1}_if_i$. Since $M^{-1}_iv_i\to 0\cdot v=0$; this yields $\tilde{f}_i(z_i)\to 0$ whenever $z_i\to z\in (\Lambda\cup\partial\Lambda)\setminus \{a\}$. On the other hand, by definition of $M_i$, we could find $z_i\in \Lambda_i\setminus B_{R/2}(a)$ such that $|\tilde{f}_i(z_i)|>\frac12$. Since $\tilde{f}_i$ are holomorphic at $\infty$ by (\ref{eq: f_at_infty}), we may assume $z_i\in B_{R_1}(a)$ for a large fixed $R_1$. Extracting a convergent subsequence $z_{i_s}$ now yields a contradiction. 
\end{proof}

\begin{proof}[Proof of Lemma \ref{lem: f_n}.]
The proof goes by induction in $n$. The base is given by $f^{(0)}(\Lambda,a,\cdot)=f(\Lambda,a,\cdot)$. For the induction step, just apply Lemma \ref{lem: remove_singularity} to $\Lambda_i\equiv \Lambda$, $f_i(z):= s^{-1}_i [f^{(n)}(\Lambda,a+s_i,z)-f^{(n)}(\Lambda,a,z)]$ and $v_i(z):=s_i^{-1}\left[\frac 1{(z-a-s_i)^{n+1}}-\frac1{(z-a)^{n+1}}\right]$, where $s_i\to 0$.
\end{proof}
In a similar manner, one obtains the following Lipschitzness result.

\begin{lemma}
\label{lem: lipshitz} 
Given a $C^2$-compact set $\Theta$ of almost-$\H$ domains, $a\in\R$, $n\geq 0$ and $r>0$, there is a constant $L=L(\Theta,a,n,r)$ such that if $\Lambda,\Lambda'\in \Theta$, $z\in\Lambda\cup\partial\Lambda$, $z'\in\Lambda'\cup\partial\Lambda'$ and $|z-a|>r$, $|z'-a|>r$, then 
$$
|f^{(n)}(\Lambda,a,z)-f^{(n)}(\Lambda',a,z')|\leq L \left(d_{C^1}(\Lambda',\Lambda)+|z'-z|\right).
$$
\end{lemma}
\begin{proof} Assume that such $L$ does not exist. Then, there exist sequences $\Lambda_i,\Lambda_i', z_i, z_i'$ satisfying the conditions of the Lemma and $p_i\to 0$ such that $p_i d_i^{-1}|f^{(n)}(\Lambda_i,a,z_i)-f^{(n)}(\Lambda'_i,a,z'_i)|$ does not tend to zero, where $d_i=d_{C^1}(\Lambda'_i,\Lambda_i)+|z_i'-z_i|$. By passing to a subsequence, we may assume that $\Lambda_i$ and $\Lambda'_i$ converge to $\Lambda$ and $\Lambda'$ in $C^2$ metric and that $z_i\to z$, $z'_i\to z'$. By Lemma \ref{lem: remove_singularity}, one has $f^{(n)}(\Lambda_i,a,z_i)\to f^{(n)}(\Lambda,a,z)$ and $f^{(n)}(\Lambda'_i,a,z'_i)\to f^{(n)}(\Lambda',a,z')$, therefore we must have $\Lambda=\Lambda'$ and $z=z'$.

Let $\Lambda''_i$ be an almost-$\H$ domain such that $\Lambda_i\subset\Lambda''_i$ $\Lambda'_i\subset\Lambda''_i$ and $d_{C^2}(\Lambda_i,\Lambda''_i)<10 d_{C^2}(\Lambda_i,\Lambda'_i)$. Then 
\begin{multline}
|f^{(n)}(\Lambda_i,a,z_i)-f^{(n)}(\Lambda'_i,a,z'_i)|\leq |f^{(n)}(\Lambda''_i,a,z_i)-f^{(n)}(\Lambda''_i,a,z'_i)|\\+|f^{(n)}(\Lambda_i,a,z_i)-f^{(n)}(\Lambda''_i,a,z_i)|+|f^{(n)}(\Lambda'_i,a,z'_i)-f^{(n)}(\Lambda''_i,a,z'_i)|
\label{eq: lip_triangle}
\end{multline}
 We will show that all three terms tend to zero when multiplied by $p_id_i^{-1}$, which yields a contradiction. Indeed, if $z\in\Lambda$ or $z\in\R\cup\infty$, then one has 
 $$
 d_i^{-1}|f^{(n)}(\Lambda''_i,a,z_i)-f^{(n)}(\Lambda''_i,a,z'_i)|\leq \frac{|f^{(n)}(\Lambda''_i,a,z_i)-f^{(n)}(\Lambda''_i,a,z'_i)|}{|z_i-z_i'|}\to |\partial_zf(\Lambda,a,z)|,
 $$
 since the uniform convergence of holomorphic functions implies the $C^{\infty}$ convergence. In the case $z\in \partial \Lambda\backslash \R$, let $l$ denote the connected component of $\C\setminus \Lambda$ containing $z$, and $l_i$ the corresponding component of $\C\setminus \Lambda''_i$. Let $\varphi_i$ be a conformal map from $\hat{\C}\setminus l_i$ to $\H$, chosen so that $\varphi_i$ converge to a conformal map $\varphi$ uniformly on compact sets of $\hat{\C}\setminus l$ and in the sense of $C^2$ up to $l$. The uniqueness in the boundary value problem (\ref{eq: obs_on_bdry}) implies the following generalization of the conformal covariance property (\ref{eq: conf_inv_f}):
 \begin{equation}
f^{(n)}(\Lambda''_i,a,\cdot)=\sum_{k=0}^n c^{(k)}_i \varphi_i'(\cdot)^{\frac12}f^{(k)}(\varphi_i(\Lambda''_i),\varphi_i(a),\varphi_i(\cdot)),  
\label{eq: conf_cov_f_n}
 \end{equation}
 where the real coefficients $c^{(k)}_i$, depending on the first $k$ derivatives of $\varphi_i$ at $a$, converge to the corresponding coefficients $c^{(k)}$. Employing this relation, Lemma \ref{lem: remove_singularity} and induction on $k$, it is easy to see that $f^{(k)}(\varphi_i(\Lambda''_i),\varphi_i(a),w_i)$ converges to $f^{(k)}(\varphi(\Lambda),\varphi(a),w)$ whenever $w_i\to w\in \Lambda\cup\overline{\Lambda}\cup \R$. Using Schwarz reflection, we see that this convergence holds in the $C^{\infty}$ sense in a neighborhood of such a $w$, which implies the bound 
 $$|f^{(k)}(\varphi_i(\Lambda''_i),\varphi_i(a),w_i)-f^{(k)}(\varphi_i(\Lambda''_i),\varphi_i(a),w'_i)|\leq C|w_i'-w_i|;$$ 
 hereinafter $C$ denotes a positive quantity that does not depend on $i$. By the assumptions on the sequence $\varphi_i(z)$, we have $|\varphi_i'(z_i)^{\frac12}-\varphi_i'(z_i')^{\frac12}|\leq C|z_i'-z_i|$ and $|\varphi_i(z_i)-\varphi_i(z_i')|\leq C|z_i'-z_i|$. Plugging this estimates into (\ref{eq: conf_cov_f_n}) gives that
 $$
 \left|f^{(n)}(\Lambda''_i,a,z_i)-f^{(n)}(\Lambda''_i,a,z'_i)\right|\leq C |z'_i-z_i|\leq C d_i,
 $$
 which concludes the proof of the required bound for the first term in (\ref{eq: lip_triangle}).
 For the second one, and similarly for the third one, note that $\hat{f}_i(\cdot):=f^{(n)}(\Lambda_i,a,\cdot)-f^{(n)}(\Lambda''_i,a,\cdot)$ satisfies (\ref{eq: obs_on_bdry}) on $\R$. Given $w\in \partial \Lambda_i\backslash \R$, pick $w'\in\pa \Lambda''_i$ such that $|w-w'|\leq d_i$ and $|n_w-n_{w'}|\leq d_i$. Then  
\begin{multline}
\left|\Im\hat{f}_i(w)\sqrt{i n_w}\right|\\=\left|\Im \left(f^{(n)}(\Lambda''_i,a,w)(\sqrt{i n_w}-\sqrt{in_{w'}})+(f^{(n)}(\Lambda''_i,a,w)-f^{(n)}(\Lambda''_i,a,w'))\sqrt{in_{w'}}\right)\right|\\ \leq 2d_i|f^{(n)}(\Lambda''_i,a,w)|+ |f^{(n)}(\Lambda''_i,a,w)-f^{(n)}(\Lambda''_i,a,w')|\leq C\cdot d_i,
\label{eq: lip_bdry_cond}
\end{multline}
with some constant $C$ independent of $i$. Hence the derivative of $\hat{h}_i=\Im\int(p_id^{-1}_i\hat{f}_i)^2$ in the tangential direction is $o(1)$ uniformly on $\pa\Lambda_i$. By maximum principle, this implies that up to uniform $o(1)$, $\hat{h}_i$ is equal to a harmonic function with constant boundary values on each boundary component. By one more extraction, dividing by the maximum of those constants if necessary, we may assume that $\hat{h}_i$ has a limit $h$. To conclude the proof, it suffices to show that this limit must be zero, which will imply $p_id^{-1}_i\hat{f}_i(\cdot)=\sqrt{2\partial_z\hat{h}_i(z_i)}\to 0$. 

Assume, on the contrary, that the limit of $h$ is non-trivial, and consider the boundary component with the maximal (constant) value of $h$. Boundary Harnack's estimate implies that there is a constant $c>0$ such that for $i$ large enough, the outer normal derivative of $\hat{h}_i$ is greater than $c$ on the corresponding component of $\Lambda_i$. Thus 
$$
c\leq p^2_id^{-2}_i\Im\left[\hat{f}^2_i(w)n_w\right]=-p^2_id^{-2}_i\left(\Re \hat{f}_i(w)\sqrt{in_w})\right)^2+p^2_id^{-2}_i\left(\Im \hat{f}_i(w)\sqrt{in_w})\right)^2, 
$$
which, together with (\ref{eq: lip_bdry_cond}), leads to a contradiction.
\end{proof}

\begin{lemma}
\label{lem: diff_one}
 Let $\epsilon>0$. Assume that $\Lambda$ is an almost-$\H$ domain, $a\in\R$, $b\in \pa\Lambda$ such that $|b-a|>10\epsilon$, and let $K\subset \overline{\H}$ be an $\H$-hull of twice half-plane capacity $t<2\epsilon^2$ and such that $K\subset \{z:|\Re (z-a)|<2\epsilon\}$. Let $g_K(z)$ denote the conformal map from $\H\setminus K$ to $\H$ with hydrodynamic normalization at infinity, $\Lambda_K:=g_K(\Lambda)$. Then
\begin{multline}
g_{K}'(z)^\frac12f(\Lambda_K,a,g_{K}(z))- f(\Lambda,a,z)=\\t\left(-3f^{(2)}\left(\Lambda,a,z\right)+\partial_zQ(\Lambda_K,a,a)f(\Lambda,a,z)\right)+O(\epsilon^3),  \label{eq: var_a}                                                                                                         
\end{multline}
and 
\begin{multline}
|g_{K}'(b)|^\frac12g_{K}'(z)^\frac12f(\Lambda_K,g_K(b),g_{K}(z))- f(\Lambda,b,z)=\\t\left(2\partial_zf(\Lambda_K,g_K(b),a)f(\Lambda,a,z)-f(\Lambda_K,g_K(b),a)f^{(1)}(\Lambda,a,z)\right)+O(\epsilon^3),                                                                                                           
\label{eq: var_b}
\end{multline}
where $O(\epsilon^3)$ is uniform over $C^2$-compact subsets of  $(\Lambda,a)$ and over $b,z$ at distance at least $r$ from $a$ for each fixed $r>10\epsilon$. 
\end{lemma}

\begin{proof}
Assume that (\ref{eq: var_a}) does not hold. This means that there are sequences $\Lambda_i$, $a_i$, $\epsilon_i$, $K_i$,  and $p_i\to 0$ such that $|z_i-a_i|>r$ and $f_i(z_i)$ does not tend to zero, where
\begin{multline}
f_i(z):=p_i\epsilon_i^{-3}\left(g_{K_i}'(z)^\frac12f(\Lambda_{K_i},a_i,g_{K_i}(z))-f(\Lambda_i,a_i,z)\right.\\ \left.+3t_i f^{(2)}\left(\Lambda_i,a_i,z\right)-t_i\partial_zQ(\Lambda_{K_i},a_i,a_i)f(\Lambda_i,a_i,z))\right).                                                                                                           
\end{multline}
By compactness, we may assume that those sequences converge to $\Lambda,a,K,z$. Then, $g_{K_i}(z_i)$ and $g_{K_i}'(z_i)$ converge to $g_K(z)$ and $g_K'(z)$, and also,  by Lemma \ref{lem: remove_singularity}, $f^{(n)}\left(\Lambda_i,a_i,z_i\right)\to f^{(n)}\left(\Lambda,a,z\right)$ and $f\left(\Lambda_{K_i},a_i,g_{K_i}(z_i)\right)\to f\left(\Lambda_K,a,g_{K}(z)\right)$. Hence, necessarily $\epsilon_i\to 0$.

To obtain a contradiction, we are going to check that $f_i$ satisfy the conditions of Lemma~\ref{lem: remove_singularity} with $v\equiv 0$. Indeed, $f_i$ obeys (\ref{eq: obs_on_bdry}) on $\partial\Lambda_i\setminus (K_i\cup\{a_i, g^{-1}_{K_i}(a_i)\})$, and the set $K_i\cup\{a_i, g^{-1}_{K_i}(a_i)\}$ shrinks to $a$. Fix $R>r>0$ so that $B_R(a_i)\cap \H$ is contained in $\Lambda_i$ for all $i$.
Standard estimates for Loewner chains (see, e.g., \cite[Eq. (3.2)--(3.4)]{Izy_free}) imply that for every fixed $r>0$, one has
$$
g_{K_i}'(z)^\frac12=1-\frac{t_i}{(z-a_i)^2}+O(\epsilon^3);\quad g_{K_i}(z)=z+\frac{2t_i}{z-a_i}+O(\epsilon^3).
$$
uniformly over $z$ such that $|z-a_i|>r$. Therefore,
\begin{multline}
 g_{K_i}'(z)^\frac12f(\Lambda_{K_i},a_i,g_{K_i}(z))\\=\left(1-\frac{t_i}{(z-a_i)^2}\right)\left(\frac{1}{g_{K_i}(z)-a_i}+Q(\Lambda_{K_i},a_i,g_{K_i}(z))\right)+O(\epsilon^3_i).
\end{multline}
We have 
$$
 \frac{1}{g_{K_i}(z)-a_i}=\frac{1}{z-a_i}-\frac{2t_i}{(z-a_i)^3}+O(\epsilon_i^3)
$$
and
\begin{multline*}
Q(\Lambda_{K_i},a_i,g_{K_i}(z))=Q(\Lambda_{K_i},a_i,z)+\frac{2t_i}{z-a_i}\partial_zQ(\Lambda_{K_i},a_i,z)+O(\epsilon^3_i)\\=\left(z-a_i+ \frac{2t_i}{z-a_i}\right)\partial_zQ(\Lambda_{K_i},a_i,a_i)+(z-a_i)^2\hat{u}_i(z)+2t_i\hat{\hat{u}}_i(z)+O(\epsilon^3_i),
\end{multline*}
where $\hat{u}_i$ and $\hat{\hat{u}}_i$ are analytic in $B_R(a_i)$; here we have used that $Q(\Lambda,a,a)=0$. Putting everything together, we obtain  
\begin{multline*}
 g_{K_i}'(z)^\frac12f(\Lambda_{K_i},a_i,g_{K_i}(z))=\frac{1}{z-a_i}-\frac{3t_i}{(z-a_i)^3}+\frac{t_i}{z-a_i}\partial_zQ(\Lambda_{K_i},a_i,a_i)+\tilde{u}_i(z)+O(\epsilon^3_i),
\end{multline*}
where $\tilde{u}_i$ is analytic in $B_R(a_i)$. Thus $f_i(z)=O(p_i)+p_i\epsilon_i^{-3}\tilde{u}_i(z)$, and, applying Lemma \ref{lem: remove_singularity}, we get that $f_i(z_i)\to 0$, the desired contradiction.

The proof of (\ref{eq: var_b}) is similar. Indeed, note that the left-hand side of \ref{eq: var_b} has no singularity at $b$, and thus satisfies (\ref{eq: obs_on_bdry}) everywhere on $\partial\Lambda\setminus K$. We have
\begin{multline}
 |g_{K}'(b)|^\frac12g_{K}'(z)^\frac12f(\Lambda_{K},g_K(b),g_{K}(z))\\=|g_{K}'(b)|^\frac12\left(1-\frac{t}{(z-a)^2}\right)\left(f(\Lambda_K,g_K(b),z)+\partial_zf(\Lambda_K,g_K(b),z)\frac{2t}{z-a}\right)+O(\epsilon^3)
 \\=\partial_zf(\Lambda_K,g_K(b),a)\frac{2t}{z-a}-f(\Lambda_K,g_K(b),a)\frac{t}{(z-a)^2}+u(z)+O(\epsilon^3),
\end{multline}
where $u(\cdot)$ is holomorphic near $a$ and we have dropped the subscript $i$ everywhere. Up to a function analytic near $a$, this is equal to the right-hand side of (\ref{eq: var_b}), and one more application of Lemma \ref{lem: remove_singularity} concludes the proof.
\end{proof}

\begin{rem}
\label{rem: improve_a_lemma}
 Using the Lipschitz estimate of Lemma \ref{lem: lipshitz} and Lemma \ref{lem: symmetry} below, one can replace $\Lambda_K$ and $g_K(b)$ in the right-hand side of (\ref{eq: exp_Delta_1}) and (\ref{eq: exp_Delta_2}) by $\Lambda$ and $b$. Indeed, the difference between $\Lambda$ and $\Lambda_K$ in $C^2$ metric is of order $\epsilon^2$, and Lemma \ref{lem: symmetry} shows that $f(\Lambda,b,a)=\sqrt{in_{b}}f(\Lambda,a,b)$ and $\partial_zf(\Lambda,b,a)=\sqrt{in_{b}}f^{(1)}(\Lambda,a,b)$, to which Lemma \ref{lem: lipshitz} readily applies.
\end{rem}

\begin{lemma}
 \label{lem: symmetry} 
For any domain $\Omega$ whose boundary is smooth near $a_1,\dots,a_{2k}$, and for any permutation $\sigma$, one has 
$$\sqrt{in_{a_{2k}}}f(\Omega,a_1,\dots,a_{2k})=(-1)^{\text{sgn}(\sigma)}\sqrt{in_{a_{\sigma(2k)}}}f(\Omega,a_{\sigma(1)},\dots,a_{\sigma(2k)}),$$ provided that for each $a_i$, the  sign of $\sqrt{in_{a_i}}$ is chosen consistently at all occurrences. 
\end{lemma}

\begin{proof}
We first establish the case $k=1$. Consider the harmonic (possibly multi-valued) function $h:=\Im \;\int f(\Omega,a_1,z)f(\Omega,a_2,z) dz$. It follows from (\ref{eq: obs_on_bdry}) and (\ref{eq: obs_at_a}) that $h$ is locally constant on the boundary, except at $a_1,a_2$, where it has jumps  of size $-\pi \sqrt{in_{a_{1}}}f_\Cvr(\Omega,a_2,a_1)$ and $-\pi\sqrt{in_{a_{2}}} f_\Cvr(\Omega,a_1,a_2)$, counted in the direction of the tangent vector oriented to have $\Omega$ on its left. Since the derivative of $h$ is well-defined on $\Omega$, $h$ itself is well-defined on the universal cover, with additive monodromy. This can only be the case if one jump is  negative of the other, i. e. $\sqrt{in_{a_{1}}}f_\Cvr(\Omega,a_2,a_1)=-\sqrt{in_{a_{2}}} f_\Cvr(\Omega,a_1,a_2)$.

For the general case, note that 
$$\sqrt{in_{a_{2k}}}f(\Omega,a_1,\dots,a_{2k})=\text{Pf}_{1\leq m,r\leq 2k}\;  \left[\sqrt{in_{a_r}}f_\Cvr(\Omega,a_m,a_r)\right]$$ 
and use the anti-symmetry of the Pfaffian. 
\end{proof}

\begin{proof}[Proof of Lemma \ref{lem: non-vanish}]
Applying a conformal map, we may assume that $\Omega$ is bounded and $\partial \Omega$ is smooth. In the following proof, we drop $\Omega$ from the notation. Note that $f_{k}(\cdot):=f_{\Cvr_Z}(a_1,\dots,a_{2k-1},\cdot)$ satisfies the following expansion as $z\to a_1$: 
\begin{equation}
f_k(z)=P_1f(a_1,z)+\sum_{r=2}^{2k-1}P_r\cdot f({a_r,a_1})+o(1)=\frac{P_1\sqrt{in_{a_1}}}{z-a_1}+o(1), 
\label{eq: exp_f_multipoint}
\end{equation}
since by Lemma \ref{lem: symmetry}, $$\sum_{r=2}^{2k-1}P_r\cdot f({a_r,a_1})=(in_{a_1})^{-\frac12}\sum_{r=2}^{2k-1}P_r\sqrt{in_{a_{1}}}f({a_1,a_r}),$$ which is a Pfaffian of a matrix with two coinciding columns. Expansions similar to (\ref{eq: exp_f_multipoint}) hold as $z\to a_2,$ \dots, $z\to a_{2k-1}$, and $f_k$ satisfies (\ref{eq: obs_on_bdry}) on $\pa\Omega\setminus \{a_1,\ldots,a_{2k-1}\}.$ Hence, the harmonic function $h=\Im \int f_k^2$ is constant on each boundary component of $\Lambda$, except for a possible blow-up as a non-negative multiple of the Poisson kernel at marked points. Let $l_{min}$ be the boundary component with the smallest value of that constant. Note that by Harnack principle, unless $h\equiv \const$, one has $\partial_n h<0$ on $l_{min}$; in particular, $f_k\neq 0$ and therefore $\res_{a_i}f_k\neq 0$, $a_i\in l_{min}$. Write 
$$
\sqrt{-\partial_n h(z)}=\sqrt{-\Im (f^2_k(z)n_z)}=\sqrt{if^2_k(z)n_z}=f_k(z) \sqrt{in_z},\quad z\in l_{min}\setminus\{a_1,\ldots,a_{2k-1}\}
$$ where we have used (\ref{eq: obs_on_bdry}) in the second equality. Note that, as $z$ goes along $l_{min}$, the left-hand side changes sign at the poles $a_i\in l_{min}$ and cannot change it in between. We conclude that $f_k$ changes sign around $l_{min}$ if and only if there is an \emph{even} number of poles $a_i \in l_{min}$. By definition of $\Cvr_Z$, this is only possible if $a_{2k}\in l_{min}$, and hence $f_k(a_{2k})\neq 0$. 

 It remains to show that $f_k(z)$ cannot be identically zero. We first claim that there exists a small perturbation $\{a'_1,\ldots,a'_{2k-1}\}$ of $\{a_1,\ldots,a_{2k-1}\}$, such that $f(a'_1,\dots,a'_{2k-1},\cdot)$ is not identically zero. Indeed, otherwise, each coefficient $P_r$ in (\ref{eq: obs_tilda}) vanishes identically as a function of marked points in a neighborhood of $\{a_1,\ldots,a_{2k-1}\}$. However, $P_r=\sqrt{in_{a_{2k-1}}}f(a_1,\dots,\widehat{a}_r,\dots,a_{2k-1})$,  so the holomorphic function $f(a'_1,\dots,\widehat{a'}_r,\dots,\cdot)$ vanishes on a boundary arc and hence is also identically zero, where $\{a_1',\ldots,a'_{2k-1}\}$ is \emph{any} small perturbation of $\{a_1,\ldots,a_{2k-1}\}$. By iterating the argument, we deduce that each $f(a_i,z)$ is identically zero, a contradiction which shows the existence of a required perturbation.

By the first part of the proof, $f(a'_1,\dots,a'_{2k-1},a_{2k})\neq 0$. Then Lemma \ref{lem: symmetry} implies that $f(a'_1,\dots,a'_{2k-2},a_{2k},\cdot)$ is not identically zero, and hence $f(a'_1,\dots,a'_{2k-2},a_{2k},a_{2k-1})\neq 0$. Iterating, one can replace each $a_i'$ by $a_i$.
\end{proof}

\begin{proof}[Proof of Lemma \ref{lem: conv_ratio}]
Assume that the uniform convergence does not hold, i. e., that there is a collection of $\delta_i\to 0$, initial segments $\gamma^{(i)}_{[0,t_i]}$ in $\Omega^{\delta_i}$ and points $z_i\in\mathcal{C}$ such that the expression in (\ref{eq: conv_ratio}) stays bounded from below by a constant. Extracting a subsequence, we may assume that $\Omega^{\delta_i}\setminus \gamma^{(i)}_{[0,t_i]}\stackrel{\text{Cara}}{\longrightarrow} \Omega'$, $\gamma^{(i)}_{t_i}\to \gamma_t$ and $z_i\to z\in \mathcal{C}$. Note that in this case, $\Omega^{\delta_i}_{n_i}\stackrel{\text{Cara}}{\to}\Omega'$ and $\gamma_{n'_i}\to \gamma_t$, where $n_i=n_{t_i}$ and $\Omega^\delta_n$, $n_i'$ are as discussed after Proposition \ref{prop: mart}. Indeed, $\Omega^\delta_n\subset\Omega^{\delta}\backslash \gamma_{[1,n]}$ and every point of their difference can be cut off by a cross-cut of length at most $2\delta$, whereas $\gamma_{n}$ and $\gamma_{n'}$ can be cut off by a common cross-cut of length at most $2\delta$.

By Corollary \ref{cor: conv_obs}, Lemma \ref{lem: non-vanish} and (\ref{eq: decouple}), one has
$$
M^{\delta_i}_{\gamma^{(i)}}(z_i) \to \frac{f(\Omega',\gamma_t,\dots,z)}{f(\Omega',\gamma_t,\dots,a_{2k})}.
$$ However,  $f(\Omega,a_1,\dots,z)/f(\Omega,a_1,\dots,a_{2k})$ is Carathéodory continuous (since the numerator and the denominator are limits of the discrete observables under Carath\'eodory convergence of domains), and hence $M_{\gamma^{(i)}}(z_i)$ converges to the same quantity as $M^{\delta_i}_{\gamma^{(i)}}(z_i)$. This gives a contradiction.
\end{proof}

\begin{proof}[Proof of Lemma \ref{lem: variation}]
Recall that by definition and conformal covariance of the observable, 
$$
M_{\gamma_{[0,t]}}(z)=g_{t}'(z)^\frac12 \sum^{2k-1}_{s=1}\tilde{P}_s \cdot f(\Lambda_t,a_s(t),g_t(z)),
$$ 
where $$\tilde{P}_s=(-1)^{s+1}\frac{\sqrt{in_{a_{2k}}}\mathrm{Pf}[\sqrt{in_{a_r(t)}}f(\Lambda_t,a_m(t),a_r(t)]_{1\leq m,r\neq s\leq 2k-1}}{|g'_t(a_{2k})|^{\frac12}\mathrm{Pf}[\sqrt{in_{a_r(t)}}f(\Lambda_t,a_m(t),a_r(t)]_{1\leq m,r\leq 2k}}$$ does not depend on $z$. Note that Lemmas \ref{lem: f_n}, \ref{lem: diff_one} and Remark \ref{rem: improve_a_lemma} allow one to expand $M_{\gamma_{[0,t_2]}}-M_{\gamma_{[0,t_1]}}$ in a Taylor series up to $O(\epsilon^3)$. For example, 
\begin{multline}
g_{t_2}'(z)^\frac12f(\Lambda_{t_2},a_1(t_2),g_{t_2}(z))-g_{t_1}'(z)^\frac12f(\Lambda_{t_1},a_1(t_1),g_{t_1}(z))\\=g_{t_1}'(z)^\frac12\left(g_{K}'(z)^\frac12f(\Lambda_{t_2},a_1(t_2),g_{t_2}(z))-f\right)\\=g_{t_1}'(z)^\frac12\left(-3\Delta_tf^{(2)}+\Delta_t\partial_zQ f+f(\Lambda_{t_1},a_1(t_2),g_{t_1}(z))-f\right)+O(\epsilon^3)\\=g_{t_1}'(z)^\frac12 \left(-3\Delta_tf^{(2)}+\Delta_t\partial_zQ f+\Delta^2_af^{(2)}+\Delta_af^{(1)}\right)+O(\epsilon^3),
\label{eq: diff_first_term}
\end{multline}
where the functions with omitted arguments are evaluated at $\Lambda_{t_1},a_1(t_1),g_{t_1}(z)$; in the second line, $K=g_{t_1}(\gamma_{[t_1,t_2]})$, in the third line we have used (\ref{eq: var_a}), modified according to Remark \ref{rem: improve_a_lemma}, and in the fourth line we use that $f^{(n)}=\frac{1}{n!}\frac{\partial^n}{\partial a^n}f$. 

The overall result of the Taylor expansion, viewed as a function of $z$, will be $g_{t_1}'(z)^\frac12$ times a linear combination of $f^{(2)},$ $f^{(1)}$, $f$ and $f(\Lambda_{t_1},a_s(t_1),g_{t_1}(z))$ for $s=2,\dots,2k-1$; we are interested in the coefficients in front of 
 $f^{(2)}$ and $f^{(1)}$. In fact, $f^{(2)}$ only appears in (\ref{eq: diff_first_term}); hence its coefficient is $R(\Delta_a^2-3\Delta_t)$. The term $f^{(1)}$ comes from (\ref{eq: diff_first_term}) with the coefficient $\Delta^2_aR'+R\Delta_a$, and also pops up when applying (\ref{eq: var_b}) to $f(\Lambda_t,a_s(t),g_t(z))$, $s=2,\ldots,2k-1$. The latter contributions sum up to zero, since the corresponding coefficient is proportional to 
 $$
\sum_{s=2}^{2k-1}P_sf(\Lambda_{t_1},a_s(t_1),a_1(t_1))=0,
 $$
 as explained in the proof of Lemma \ref{lem: non-vanish} after the equation (\ref{eq: exp_f_multipoint}). The property $c_i\in \sqrt{in_{a_{2k}}}\cdot\R$ follows readily from the same property of $\tilde{P_s}$.
 \end{proof}

\begin{proof}[Proof of Lemma \ref{lem: coef_vanish}]
 Assume, on the contrary, that such an $L$ does not exist. Then there is a sequence of domains $(\Lambda^{(i)},a_1^{(i)}\dots,a_{2k-1}^{(i)})\in \CSpace(\crs)$ and a sequence of spinors $\Phi_i$ in those domains defined by (\ref{eq: Phi}) with constants $C_{1,2}^{(i)}$, $c_m^{(i)}$, $1\leq m\leq{2k-1}$ such that 
$$
\max \{C_{1}^{(i)}, C_2^{(i)}\} \to \infty;
$$
but $\max_{g^\delta_{t_i}(\mathcal{C})} {|\Phi_i|}$ is bounded. Denote $M_i:=\max\{C_{1,2}^{(i)};\,c_m^{(i)}\}$. By extracting a subsequence, we may assume that  $(\Lambda^{(i)},a_1^{(i)}\dots,a_{2k-1}^{(i)}) \to (\Lambda,a_1,\dots,a_{2k-1})$ and $M^{-1}_i\Phi_i\to \Phi$, where $\Phi$ is also of the form (\ref{eq: Phi}), with one of the coefficients equal to one, and hence it is not identically zero. On the other hand, $\Phi$ vanishes on a set with non-zero interior, which is a contradiction. 
\end{proof}

\begin{proof}[Proof of Lemma \ref{lem: generic}.]
 We have to show that for each realization of $\gamma_{[0,t]}$, there is at least one $s=2,\ldots,2k$ such that 
 $$
 \mathrm{Pf}[\sqrt{in_{a_r}}f(\Lambda^\delta_t,a^\delta_m(t),a^\delta_r(t))]_{2\leq m,r\neq s\leq 2k}\neq 0;
 $$ then the desired bound follows from continuity and compactness. However, these are coefficients in the decomposition (\ref{eq: obs_tilda}) of $f_\Cvr(\Lambda^\delta_t,a^\delta_{2k}(t),a^\delta_2(t),\dots,a^\delta_{2k-1}(t),\cdot)$ as a linear combination of $f(\Lambda^\delta_t,a^\delta_s(t),\cdot)$. By Lemma \ref{lem: non-vanish}, the former does not vanish identically, hence at least one of those coefficients is non-zero.
\end{proof}

\begin{proof}[Proof of Lemma \ref{lem: D-lipshitz}]
Note that the denominator in the definition (\ref{eq: def_D_bis}) of $D$ does not vanish by Lemma \ref{lem: non-vanish}, and thus by compactness, the corresponding denominators of $D((a_1)_{[0,t]})$ and $D((\tilde{a}_1)_{[0,\tilde{t}]})$ are bounded away from zero. Using symmetry (Lemma \ref{lem: symmetry}), we therefore see that it suffices to prove the Lipschitz estimates separately for $f^{(1)}(\Lambda,a_1,a_m)$, $f(\Lambda,a_1,a_m)$ and $f(\Lambda,a_r,a_m)$, $2\leq m\neq r \leq 2k$. 

Note that the Loewner maps $g_{t}(z)$ are uniformly Lipschitz (with their derivatives) with respect to the driving process, i. e. for any cross-cut $\crs'$ such that $\crs$ separates $\crs'$ from $a_1$, there are constants $L_0$, $L_1,\dots,$ such that 
\begin{eqnarray*}
|g_{t}(z)-\tilde{g}_{\tilde{t}}(z)|<L_0\left(|t-\tilde{t}|+\sup\limits_{0\leq s\leq \tau}|a_1(s)-\tilde{a}_1(s)|\right),
\\                                                                                                                                                                                                                          
|\partial^{(n)}_zg_{t}(z)-\partial^{(n)}_z\tilde{g}_{\tilde{t}}(z)|<L_n\left(|t-\tilde{t}|+\sup\limits_{0\leq s\leq \tau}|a_1(s)-\tilde{a}_1(s)|\right),                                                                                                                                                                                                                                       
\end{eqnarray*}
where $\tau= t\wedge \tilde{t}$ and  $z$ is separated from $\crs$ by $\crs'$ (see, e. g., \cite[Lemma 5.5]{Zhan-LERW}). Hence, it suffices to show that the functions in question are Lipschitz with respect to the metric on tuples $(\Lambda,a_1,\dots,a_{2k})$ given by $d_{C^2}(\Lambda',\Lambda'')+\max\;|a'_m-a''_m|$. For $\partial f^{(1)}(\Lambda,a_1,a_m)$, $f(\Lambda,a_1,a_m)$, and $f(\Lambda,a_r,a_m)$, where $a_r\in\R$, this is proven in Lemma \ref{lem: lipshitz}. In the remaining case $f(\Lambda,a_r,a_m)$, where neither $a_m\in \R$, nor $a_r\in\R$, it can be proven in a similar way or by conformal covariance; we leave the details to the reader.
\end{proof}

\label{technical}

\section{Proof of Propositions \ref{prop: ann1} -- \ref{prop: spin}}

In this section, we describe the proofs of Propositions \ref{prop: ann1}--\ref{prop: spin}. This includes three additional ingredients as compared to Theorem \ref{thm: conv_int}. First, we explain how to transfer the techniques developed above from half-plane Loewner chains to annulus ones. We will only describe the analog of the heuristic computation of Section \ref{sec: DP_general}; the rigorous justification goes along the same lines as in the proof of Theorem \ref{thm: conv_int}. Note that in the doubly connected case, most of the technical lemmas follow directly from the explicit formulae for the observables. We also point out that essentially, the computation in the annulus case is the same as in the radial one. The error of approximation of the annulus Loewner kernel by the radial one does not enter into the answer.

Second, we prove Proposition \ref{prop: spin} and treat the case of the target point inside the domain (or, in other words,  on a ``microscopic'' boundary component consisting of one face). Finally, we compute the partition functions explicitly. 

\subsection{Annulus Loewner chains}
\label{sec: annulus}
We start by recalling the definition of the annulus Loewner chains. The \emph{annulus Schwarz kernel} is defined to be the unique analytic function $S^p(z)$ in $\A_p$ satisfying the boundary conditions $\Re S^p(z)\equiv 0$ for $|z|=1$, $z\neq 1$, $\Re S^p(z)=\const$ for $|z|=e^{-p}$ and the expansion $S^p(z)=\frac{1+z}{1-z}+o(1)$ as $z\to 1$. Explicitly, $S^p(z)$ is given by 
 $$
 S^p(z)=v.p. \sum_{k=-\infty}^{\infty}\frac{e^{2pk}+z}{e^{2pk}-z}.
 $$
 The \emph{annulus Loewner kernel} is defined to be $V^p_\theta(z)=zS^p(ze^{-i\theta})$.  The \emph{annulus Loewner equation} with a driving function $a(t)=e^{i\theta(t)}$ reads 
 $$
 \partial_t g_t(z)=V^{p-t}_{\theta(t)}(g_t(z)),
 $$
 $g_t$ being a conformal map from a sub-domain $\Omega_t\subset \Omega$ to $\A^{p-t}$; the initial conditions are given by a map $g_0:\Omega\to \A^p$

Recall the setup of Section \ref{sec: DP_general}, bearing in mind that now $g_t$ denotes the annulus Loewner maps and not the half-plane ones; the marked points $a_m(t)$ will be parametrized by $a_m(t)=\rho_m(t) e^{i\theta_m(t)}$, where $\rho_m(t)\equiv 1$ for $a_m(t)$ on the outer boundary and $\rho_m(t)\equiv e^{t-p}$ for $a_m(t)$ on the inner one.

First, we need an expansion of $f(\A^p,a,z)$ near a boundary point $a$. Denote $\varphi_a(z):=i\frac{a-z}{a+z}$, which is a conformal map from the unit disc to the upper half-plane. Using (\ref{eq: conf_inv_f}) and (\ref{eq: obs_expansion}), we get 
$$
f(\A^p,a,z)=|\varphi'_a(a)|^{\frac12}\varphi'_a(z)^\frac12\left(\frac{1}{\varphi_a(z)}+o(1)\right)=\frac{\sqrt{ai}}{z-a}+o(1),\quad z\to a.
$$
Due to the computation (\ref{eq: exp_f_multipoint}), the same expansion holds for $f(\A^p,a_1,\ldots,a_{2k-1},z)$. Therefore, (\ref{eq: mart_expansion}) can be now replaced with 
\begin{equation*}
M_{\gamma_{[0,t]}}(z)=(g_t'(z))^\frac12\left(\frac{\sqrt{ie^{i\theta_1(t)}}}{g_t(z)-e^{i\theta_1(t)}}R_{\gamma_{[0,t]}}+Q_{\gamma_{[0,t]}}\left(g_t(z)\right)\right), 
\end{equation*}
 where $Q_{\gamma_{[0,t]}}$ is analytic in a neighborhood of and vanishes at $e^{i\theta_1(t)}$, and $R_{\gamma_{[0,t]}}$ is given by (\ref{eq: R}) with $\Lambda_t$ replaced by $\A^{p-t}$. Note that 
 $$
 dg_t(z)=V^{p-t}_{\theta_1(t)}(g_t(z))dt= \left(\frac{2e^{i\theta_1(t)}}{e^{i\theta_1(t)}-g_t(z)}-3e^{i\theta_1(t)}+o(1)\right)dt, \quad g_t(z)\to e^{i\theta_1(t)};
 $$
 $$
 dg'_t(z)^\frac12 =\frac12g'_t(z)^\frac12(V^{p-t}_{\theta_1(t)})'(g_t(z))dt= g'_t(z)^\frac12\left(\frac{e^{2i\theta_1(t)}}{(e^{i\theta_1(t)}-g_t(z))^2}+O(1)\right), \quad g_t(z)\to e^{i\theta_1(t)},
 $$
 and that 
 $$de^{i\theta_1(t)}=ie^{i\theta_1(t)}d\theta_1(t)-\frac12e^{i\theta_1(t)}d[\theta_1]_t;$$ 
 $$d[e^{i\theta_1(t)}]_t=-e^{i2\theta_1(t)}d[\theta]_t.$$
Taking the It\^o derivative of $M_{\gamma_{[0,t]}}$ with the help of these formulae, one obtains
$$
dM_{\gamma_{[0,t]}}(z)=(g_t'(z))^\frac12\left( \frac{\sqrt{i}e^{\frac52i\theta_1(t)}(3dt-d[\theta_1]_t)R}{(g_t(z)-e^{i\theta_1(t)})^3}+\frac{i^\frac32e^{\frac32i\theta_1(t)}(Rd\theta_1(t)+d[\theta_1,R]_t)}{(g_t(z)-e^{i\theta_1(t)})^2}+\dots\right),
$$
from which one readily gets $d\theta_1(t)=\sqrt{3}dB_t-3\partial_{\theta_1}\log R$, as claimed in Propositions \ref{prop: ann1} and \ref{prop: ann2}.

\subsection{Change of normalization and a single-point boundary component}
We begin by explaining a proof of Proposition \ref{prop: spin}.
\begin{proof}[Proof of Proposition \ref{prop: spin}]
Let $\tilde{Z}(\Od,a^\delta_1,\dots,a^\delta_{2k})$ and $Z(\Od,a^\delta_1,\dots,a^\delta_{2k})$ denote the partition functions of the model with fixed spin and monochromatic boundary conditions respectively. By the same argument as in Proposition \ref{prop: mart},
$$
\tilde{M}^\delta_{\gamma_{[0,n]}}(z):=M^\delta_{\gamma_{[0,n]}}(z)\cdot
\frac{Z(\Od\setminus \gamma_{[0,n]},\gamma_n,\dots,a^\delta_{2k})}{\tilde{Z}(\Od\setminus\gamma_{[0,n]},\gamma_n,\dots,a^\delta_{2k})}
=
\sqrt{in_{a_{2k}}}\cdot\frac{F(\Od\setminus\gamma_{[0,n]},\gamma_n,\dots,z)}{\tilde{Z}(\Od\setminus\gamma_{[0,n]},\gamma_n,\dots,a^\delta_{2k})}
$$ 
is a martingale under the probability measure on $\gamma_{[0,n]}$ corresponding to the fixed boundary conditions. By our assumptions, $\tilde{M}^\delta_{\gamma_{[0,n]}}(z)$ converges to 
$
M_{\gamma_{[0,n]}}(z)\Gamma^{-1}(\Omega\setminus\gamma_{[0,t]},\gamma_t,\ldots,a_{2k}). 
$ Multiplication by $\Gamma^{-1}$, which is independent of $z$, does not affect the proof of Theorem~\ref{thm: conv_int}, except that $R_{\gamma_{[0,t]}}$ gets replaced by $\tilde{R}_{\gamma_{[0,t]}}=R_{\gamma_{[0,t]}}\Gamma^{-1}$, leading to the corresponding change of its logarithmic derivative in the final result.
\end{proof}

We now recall the definitions of the observable $f(\Omega,a,z)$ in the cases when there are single point boundary components and/or $a\in \Omega$. In the former case, if $\{b\}$ is such a component, and $\Cvr$ branches around $b$, the condition (\ref{eq: obs_on_bdry}) becomes \cite{ChIz}
\begin{equation}
\label{eq: obs_at_b}
\lim_{z\to b} \sqrt{i}\sqrt{z-b}f(\Omega,a,z)\in i\R; 
\end{equation}
in particular, we require that the limit exists. If $a\in\Omega$, then \cite{CHHI} the condition (\ref{eq: obs_at_a}) has to be replaced with 
\begin{equation}
\label{eq: obs_at_a_inside}
\lim_{z\to a} \sqrt{i}\sqrt{z-b}f(\Omega,a,z)=1; 
\end{equation}
note that the notation in \cite{CHHI} differs from one in the present paper and in \cite{ChIz} by an extra factor of $\sqrt{i}$ in front of all the observables. 

The proof of Theorem \ref{thm: mainconv_int_1} in \cite{ChIz} does cover the case of a single-face boundary component; however, the proof of convergence of interfaces needs a modification due to the singularity of $f$ at $b$. The following lemma allows one to fix the issue. Let $b^\delta$ be an inner vertex of $\Omega^\delta$, and denote by $\{b^\delta\}$ one of its adjacent faces. Denote by $Z(\Od,a^\delta,b^\delta)$ the partition function of the Ising model in $\Od$ with interface starting at $a^\delta$ and ending at $b^\delta$, and let, as before, $\Cvr=\Cvr_Z$ stand for the double cover that branches only around $\{b^\delta\}$ and the boundary component that contains $a^\delta$. 

\begin{lemma} There is a normalizing factor $\vartheta(\delta)$ (which depends only on $\delta$), such that 
as $\delta\to 0$, one has  
\label{lem: conv_single_point}
\begin{equation}
\frac{F_\Cvr(\Od\backslash \{b^\delta\},a^\delta,\cdot)}{\vartheta(\delta) Z(\Omega^\delta,a^\delta,b^\delta)}\to
\frac{f_\Cvr(\Omega\backslash \{b\},a,\cdot)}{Z(\Omega,a,b)}. 
\end{equation}
where $Z(\Omega,a,b)=\sqrt{-i\res_b f^2_\Cvr(\Omega\backslash \{b\},a,\cdot)}$.
\end{lemma}
\begin{proof}
 This lemma was proven in \cite[Theorem 2.19]{CHHI}, in a different notation. The function $F_{[\Od,a,b]}(z)$ in the notation of the present paper corresponds to the ratio $\frac{F_\Cvr(\Od\backslash \{b^\delta\},a^\delta,z^\delta)}{F_\Cvr(\Od\backslash \{b^\delta\},a^\delta,a^\delta)}$, $F_{[\Od,a,b]}(b+\frac{\delta}{2})=\frac{Z(\Od,a^\delta,b^\delta)}{F_\Cvr(\Od\backslash \{b^\delta\},a^\delta,a^\delta)}$ and $\mathcal{B}_{[\Omega,a,b]}=Z(\Omega,a,b)$. Hence 
$$
\frac{F_\Cvr(\Od\backslash \{b^\delta\},a^\delta,\cdot)}{\vartheta(\delta) Z(\Omega^\delta,a^\delta,b^\delta)}=\frac{F_{[\Od,a,b]}(\cdot)}{\vartheta(\delta )F_{[\Od,a,b]}(b+\frac{\delta}{2})},
$$
and thus the result follows from \cite[Theorems 2.15, 2.19]{CHHI}. Note that although \cite{CHHI} only treats the case when $a^\delta$ is inside the domain, in the proof of Theorem 2.19 one deals with local behavior of $F$ near the singularity at $b$. Thus the same proof applies to the case when $a^\delta$ is on the boundary. 
\end{proof}

We will also need the following analog of Lemma \ref{lem: non-vanish}.
\begin{lemma}
\label{lem: symmetry_bis}
 If $a\in \partial\Omega$ and $b\in\Omega$, then 
$$
Z(\Omega,a,b) \neq 0. 
$$
\end{lemma}
\begin{proof}
By conformal invariance, we may assume the boundary of $\Omega$ to be smooth. Consider the function $H(w)=\Im \int^w f(\Omega\setminus \{b\},a,z)f(\Omega,b,z)dz$. Since it is an integral of a single-valued analytic function, it is well-defined up to an additive monodromy. The boundary conditions (\ref{eq: obs_on_bdry}) and (\ref{eq: obs_at_b})--(\ref{eq: obs_at_a_inside}) imply that $H$ is locally constant along the boundary, except at $a$, where it has jump $\pi f(\Omega,b,a)$, and at $b$, where it behaves as $Z(\Omega,a,b)\Im\,\log(z-b)$. The monodromies must coincide, which implies that
\begin{equation}
\label{eq: symmetry_bis}
2Z(\Omega,a,b)=f_{\Cvr_z}(\Omega,b,a). 
\end{equation}
Non-vanishing of the latter quantity is proven as in Lemma \ref{lem: non-vanish}: the harmonic function $h:=\Im \int f(\Omega,b,z)^2dz$ is constant on each boundary component, hence by Harnack principle its gradient cannot vanish on the boundary component with a minimal value of the constant, and the definition of $\Cvr_Z$ then implies that $a$ belongs to that component.
\end{proof}

\begin{proof}[Proof of Proposition \ref{prop: ann2}] 
The quantity 
$$
M^\delta_{\gamma_{[0,n]}}(z):=\frac{F_\Cvr\left(\Od\setminus \left(\{b^\delta\}\cup\gamma_{[0,n]}\right),\gamma_n,z\right)}{\vartheta(\delta) Z\left(\Od\setminus \left(\{b^\delta\}\cup\gamma_{[0,n]}\right),\gamma_n,b^\delta\right)}
$$
is a martingale with respect to the interface $\gamma_{[0,n]}$; hence, in view of Lemmas \ref{lem: conv_single_point}--\ref{lem: symmetry_bis}, the proof of Theorem \ref{thm: conv_int} readily extends, with $M_{\gamma_{[0,t]}}(z)=f(\Omega\setminus \gamma_{[0,t]},\gamma_t,z)Z(\gamma_{[0,t]},\gamma_t,b)^{-1}$ and $R_{\gamma_{[0,t]}}(z)=Z(\gamma_{[0,t]},\gamma_t,b)^{-1}$.
\end{proof}

\subsection{Explicit computation of partition functions}
\label{DP_explicit}

Let $f_{0,1}(\A_p,a,z)$ be two observables for the annulus $\A_p$ corresponding to the trivial and non-trivial covers. Our goal is to find explicit expressions for $f_{0,1}$. Consider the covering map $\psi:z\mapsto e^{iz}$ of   $\A_p$ by the strip $\bS_p:=\{0<\Im \,z< p\}$. The functions $f_{0,1}(\bS_p,\theta,w):=(\psi'(w))^{\frac12}f_{0,1}(\A_p,e^{i\theta},\psi(w))$ are holomorphic in $w\in\bS_p$, anti-periodic (respectively, periodic) with period $2\pi$ and have simple poles of residue $\pm 1$ at $\theta+2\pi m$, $m\in \Z$. Note also that $f_{0,1}(\bS_p,\theta,w)\in \R$ for $w\in \R\setminus \theta$ and $f_{0,1}(\bS_p,\theta,w)\in i\R$ for $w\in \R+ip$. This allows one to continue these functions by Schwarz reflection as $2pi$-antiperiodic meromorphic functions in $w\in\C$.  
These conditions characterize the functions uniquely as Jacobian elliptic functions \cite{Spec_functions} $ds(w|\pi,p)$ and $cs(w|\pi,p)$. To evaluate $f(\A_p,a_m,a_n)$, where $a_m$ and $a_n$ are on different boundary components, which corresponds to evaluation of $f_{0,1}(\bS_p,\theta,w)$ at $w$ with $\Im w=p$, we have used a transformation formula for a shift by a quarter-period. We also dropped a common factor from the final result.

In the case the target point $b$ is inside the domain, by Lemma \ref{lem: symmetry_bis}, it is enough to compute $f(\A_p,b,a)$. Passing again to the strip, consider the function $$f^2(\bS_p,\theta+i\rho,w):=\psi'(w)f^2(\A_p,\psi(\theta+i\rho),\psi(w)),$$ which is analytic and $2\pi$-periodic in $w$, purely real on $\R$ and $\R+ip$ and has  simple poles at $\theta+i\rho+2\pi m$. Schwarz reflection then allows one to extend it to an analytic function in the full plane which is also both $2\pi$- and $2pi$-periodic and has two simple poles per fundamental domain. Hence, it also has two zeros per fundamental domain. Since $f(\bS_p,\theta+i\rho,w)$ is itself purely imaginary and $2\pi$-antiperiodic along $\R+ip$, it must have zeros there, which by simple symmetry considerations have to be at $\theta+\pi+ip+2\pi\Z$. Consequently, $f^2(\bS_p,\theta+i\rho,w)$ has a double zero there, and hence
$$f^{-2}(\bS_p,\theta+i\rho,w)=c_1\wp(w-\theta-\pi-ip|2\pi,2p)+c_2.$$ The constants are determined from the condition (\ref{eq: obs_at_a_inside}), leading to the expression (\ref{eq: pf_annulus_radial}).

\def\cprime{$'$}

\end{document}